\documentclass{article}

\usepackage{amsmath,amsfonts,amssymb,amsthm,bbm}
\usepackage{xspace,mathrsfs,hyperref}
\usepackage[linesnumbered,ruled,vlined]{algorithm2e}
\usepackage[figurename=Fig.,labelfont=bf,labelsep=period]{caption}
\usepackage[margin=1in]{geometry}
\usepackage{graphicx}
\usepackage{enumerate}
\usepackage[dvipsnames]{xcolor}

\hypersetup{
    colorlinks = true,
    linkcolor = black,
    citecolor = olive,
    urlcolor = cyan,
}

\newcommand{\calV}{{\mathcal V}}
\newcommand{\calG}{{\mathcal G}}
\newcommand{\calE}{{\mathcal E}}
\newcommand{\calS}{{\mathcal S}}
\newcommand{\calT}{{\mathcal T}}
\newcommand{\calL}{{\mathcal L}}
\newcommand{\calU}{{\mathcal U}}
\newcommand{\calP}{{\mathcal P}}

\newcommand{\calJ}{{\mathcal J}}
\newcommand{\calM}{{\mathcal M}}

\newcommand{\oball}{B^o}
\newcommand{\pt}{{(t)}}
\newcommand{\pone}{{(1)}}
\newcommand{\ptwo}{{(2)}}

\newcommand{\E}{{\mathbb E}}

\newcommand{\opt}{\mathsf{OPT}}
\newcommand{\sol}{\mathsf{SOL}}
\newcommand{\etal}{\textit{et~al.}\xspace}

\newcommand{\lnode}{\mathscr{L}}
\newcommand{\rnode}{\mathscr{R}}
\newcommand{\tw}{\widetilde{w}}
\newcommand{\hw}{\widehat{w}}
\newcommand{\bw}{\overline{w}}

\newcommand{\source}{\mathsf{s}}
\newcommand{\sink}{\mathsf{t}}
\newcommand{\dav}{d_\mathsf{av}}
\newcommand{\dhav}{{d_{-h,\mathsf{av}}}}
\newcommand{\dhfar}{{d_{-h,\mathsf{far}}}}

\newcommand{\LP}{\mathsf{LP}}
\newcommand{\mfl}{\textsf{MFL}\xspace}
\newcommand{\tmfl}{\textsf{TM-MFL}\xspace}

\newcommand{\network}{\mathcal{N}}
\newcommand{\floor}[1]{\lfloor #1\rfloor}
\newcommand{\ceil}[1]{\lceil #1\rceil}
\newcommand{\topl}{\mathrm{Top}_\ell}

\newcommand{\vco}{\vec{o}}
\newcommand{\vcv}{\vec{v}}
\newcommand{\fracflow}{\tilde{f}}
\newcommand{\intflow}{\bar{f}}

\DeclareMathOperator*{\argmin}{arg\,min}

\newcommand{\ckc}{\textsf{CkC}\xspace}
\newcommand{\dks}{\textsf{DkSup}\xspace}
\newcommand{\rdks}{\textsf{DkSupOut}\xspace}
\newcommand{\dom}{\textsf{DOkMed}\xspace}

\newtheorem{theorem}{Theorem}
\newtheorem{lemma}[theorem]{Lemma}
\newtheorem{corollary}[theorem]{Corollary}

\theoremstyle{definition}
\newtheorem{definition}{Definition}

\theoremstyle{remark}
\newtheorem{remark}{Remark}

\title{Approximation algorithms for clustering with dynamic points}

\author{Shichuan Deng$^1$ \and Jian Li$^1$ \and Yuval Rabani$^2$}

\date{$^1$Tsinghua University, China \\ $^2$The Hebrew University of Jerusalem, Israel}

\begin{document}

\maketitle

\begin{abstract}
We study two generalizations of classic clustering problems called \emph{dynamic ordered $k$-median} and 
\emph{dynamic $k$-supplier}, where the points that need clustering evolve over time, and we are allowed to move the cluster centers between consecutive time steps. 
In these dynamic clustering problems, the general goal is to minimize certain combinations of the \emph{service cost} of points and the \emph{movement cost} of centers, or to minimize one subject to some constraints on the other. We obtain a constant-factor approximation algorithm for dynamic ordered $k$-median under mild assumptions on the input. We give a 3-approximation for dynamic $k$-supplier and a multi-criteria approximation for its outlier version where some points can be discarded, when the number of time steps is two. We complement the algorithms with almost matching hardness results.
\end{abstract}


\section{Introduction}

Clustering a data set of points in a metric space is a fundamental abstraction of many practical problems of interest 
and has been subject to extensive study as a fundamental problem of both machine learning and combinatorial optimization. 
In particular, cluster analysis is one of the main methods of unsupervised learning, and clustering often models
facility location problems.
More specifically, some of the most well-studied clustering problems involve the following
generic setting. We are given a set $C$ of points in a metric space, and our goal is to compute a set of $k$ centers that 
optimizes a certain objective function which involves the distances between the points in $C$ and the computed centers. 
Two prominent examples are the $k$-median problem and the $k$-center problem. They are formally defined as follows.
Let $S$ denote the computed set of $k$ cluster centers, $d(j,S)=\min_{i\in S}d(i,j)$ be the minimum distance from 
a point $j\in C$ to $S$, and $D=(d(j,S))_{j\in C}$ be the \emph{service cost vector}.
The $k$-median problem aims to minimize the $L_1$ objective $\|D\|_1=\sum_{j\in C}d(j,S)$ over the choices of $S$, 
and $k$-center aims to minimize the $L_\infty$ objective $\|D\|_\infty =\max_{j\in C}d(j,S)$.
In general metric spaces and when $k$ is not a fixed constant, both problems are APX-hard \cite{hsu1979easy,jain2002greedy} and exhibit 
constant factor approximation algorithms \cite{gonzalez1985clustering,hochbaum1985best,arya2004local,charikar2012dependent,li2016pseudo,byrka2017improved}.
An important generalization is the \emph{ordered $k$-median} problem. Here, in addition to $C$
and $k$, we are given also a non-increasing weight vector $w\in\mathbb{R}_{\geq0}^{|C|}$. Letting $D^\downarrow$ denote 
the sorted version of $D$ in non-increasing order, the objective of ordered $k$-median is to minimize $w^\top D^\downarrow$.
This problem generalizes both $k$-center and $k$-median and 
has attracted significant attention recently. Several constant
factor approximation algorithms have been developed \cite{aouad2019ordered, byrka2018constant,chakrabarty2018interpolating,chakrabarty2019approximation}. We note that in the facility location literature, points are called clients and centers are called
facilities, and we will use these terms interchangeably.

In this paper, we study several dynamic versions of the classic clustering problems, in which the points that need clustering may change for each time step, and we are allowed to move the cluster centers in 
each time step, either subject to a constraint on the distance moved, or by incurring a cost proportional to 
that distance. These versions are motivated in general by practical applications of clustering,  where the 
data set evolves over time, reflecting an evolution of the underlying clustering model. Consider, for instance,
a data set representing the active users of a web service, and a clustering representing some meaningful
segmentation of the user base. The segmentation should be allowed to change over time, but if it is changed
drastically between time steps, then it is probably meaningless. For a more concrete example, consider the
following application scenario. A giant construction company with several construction teams is working 
in a city. The company has $k$ movable wireless base stations for their private radio communication, and each team also has a terminal device. The teams need to put their devices at a certain energy level, in order to maintain the communication channel between the device and the nearest base station. A team may finish their current project and move to another place to carry out new tasks at some time. Note that the wireless base stations are also movable at a certain expense.
Our high level objective is to have all teams covered by the base stations at all times, meanwhile minimizing the energy cost of all teams plus the cost of moving these base stations.

We study two problems of this flavor. The first problem, called dynamic ordered $k$-median (\dom), is a very general model that captures a wide range of dynamic clustering problems where the objective is to minimize the sum of service cost and movement cost. In particular, it generalizes dynamic versions of $k$-center and $k$-median. The problem is defined as follows. We are given a metric space and $T$ time steps. In each time step $t$, there is a set $C_t$ of clients 
that we need to serve, and we can choose the locations for $k$ mobile facilities to serve 
the clients (each client is served by its closest facility). Our goal is to minimize the total ordered service cost 
(i.e., the ordered $k$-median objective), summed over all time steps, plus the total cost of moving the $k$ facilities. We define the problem formally as follows.
\begin{definition}
	(\dom). Given a metric space $(X,d)$, an instance of $T$-\dom, $T\in\mathbb{Z}_+$ is specified by $\{C_t\}_{t=1}^T$, $\{F_t\}_{t=1}^T$, non-increasing weight vectors $\{w_t\in\mathbb{R}_{\geq0}^{|C_t|}\}_{t=1}^T$, and $\gamma>0,\,k\in\mathbb{Z}_+$. $T\geq2$ and $\gamma$ are constants. $C_t\subseteq X$ is the set of clients at time $t$, and $F_t\subseteq X$ is the set of candidate locations where we can place facilities. We are required to compute a sequence of multi-sets of facilities $\{A_t\}_{t=1}^T$ with $A_t\subseteq F_t,\,|A_t|=k$, minimizing the following sum of ordered service cost and movement cost,
	\begin{equation*}
	\sum_{t=1}^T w_t^\top\left(d(j,A_t)\right)_{j\in C_t}^\downarrow+\gamma\cdot\sum_{t=1}^{T-1}m(A_t,A_{t+1}),
	\end{equation*}
	where for a vector $\vec{v}$, $\vec{v}^\downarrow$ is the non-increasingly sorted version of $\vec{v}$, and \[m(Y,Z)=\min_{M_0\in M(Y,Z)}\sum_{(i,i')\in M_0}d(i,i')\] 
	is the minimum total distance among perfect matchings between two equal-sized multi-sets $Y$ and $Z$.
\end{definition}

In \dom, the second term in the objective (i.e., the total distance traveled by all facilities) can also be motivated by the online optimization problems of $k$-server \cite{manasse1990competitive,koutsoupias1995server} and $k$-taxi \cite{coester2019online,buchbinder2021online}.
In these problems, $k$ servers are present in a metric space; at each time step, one client is revealed with a request, and some server needs to travel in the metric space to fulfill the client's request, incurring a movement cost.
The unweighted version of \dom (i.e., each $w_t$ is an all-one vector) can be roughly regarded as an offline version of $k$-server, except that many clients show up at each time step, and the servers need to fix their locations and serve all clients simultaneously.

\dom is also related to the stochastic $k$-server problem, first studied by Dehghani~\etal\cite{dehghani2017stochastic}.
In this problem, we have $T$ time steps and $T$ distributions $\{P_t\}_{t\in[T]}$ given \emph{in advance}. 
The $t$-th client is drawn from $P_t$, and we can use $k$ movable servers. 
One variant they consider is that, after a client shows up, its closest server goes to the client's location and comes back, and the optimization objective is the total distance traveled by all servers; in expectation, this objective is the same as in \dom, if we consider \emph{non-ordered and weighted} clients and all weights sum up to 1 for each time slot.
We will further discuss the relation between \dom and stochastic $k$-server in Section~\ref{section:related:work}.

It is also natural to formulate dynamic clustering problems where the objective is to minimize just the service cost, subject to some constraints on the movement cost. This turns out to be technically very different from \dom. Our second problem, which we call dynamic $k$-supplier (\dks), is such a concrete problem, where the service cost is the 
$k$-supplier objective, i.e., the maximum client-facility connection distance over all time steps, and the constraint is that any facility cannot be moved further than a fixed distance $B\geq0$ between any two consecutive time steps. More formally:

\begin{definition}
	(\dks). Given a metric space $(X,d)$, an instance of $T$-\dks, $T\in\mathbb{Z}_+$ is specified by $\{C_t\}_{t=1}^T$, $\{F_t\}_{t=1}^T$, and $B\geq0,\,k\in\mathbb{Z}_+$. $T\geq2$ is the number of time steps, 
	$C_t\subseteq X$ is the set of clients for time $t$, and $F_t\subseteq X$ is the set of candidate locations where we can place facilities.
	We are required to compute a sequence of multi-sets of facilities $\{A_t\}_{t=1}^T$, with $A_t\subseteq F_t,\,|A_t|=k$, minimizing the maximum service cost of any client, $\max_t\max_{j\in C_t}d(j,A_t)$, subject to the constraint that there must exist a perfect matching between $A_t$ and $A_{t+1}$ for each $t\in[T-1]$, and the distance between each matched pair is at most $B$.

	In the outlier version (\rdks), we are additionally given the outlier constraints $\{l_t\in\mathbb{Z}_+\}_{t=1}^T$. We are required to identify a sequence of multi-sets of facilities $\{A_t\}_{t=1}^T$ and a sequence of subsets of served clients $\{S_t\subseteq C_t\}_{t=1}^T$, with $A_t\subseteq F_t,\,|A_t|=k,\,|S_t|\geq l_t$. The goal is to minimize the maximum service cost of any served client, $\max_t\max_{j\in S_t}d(j,A_t)$, with the constraint that there must exist a perfect matching between $A_t$ and $A_{t+1}$ for each $t\in[T-1]$, and the distance between each matched pair is at most $B$.
\end{definition}

For the outlier problem $T$-\rdks, a \emph{multi-criteria $(\alpha_0,\alpha_1,\dots,\alpha_T)$-approximation} is a polynomial-time algorithm that always outputs a solution with objective value at most $\alpha_0$ times the optimum, while the number of served clients is at least $\alpha_tl_t$ at time $t$, $t\in[T]$. Because standard $k$-supplier is APX-hard and prohibits any polynomial time approximation schemes (PTAS) unless $\mathrm{P=NP}$, we are interested in obtaining multi-criteria $(\alpha_0,1-\epsilon,\dots,1-\epsilon)$-approximation algorithms for some non-trivial $\alpha_0$ and any constant $\epsilon>0$, or pure approximations where $\alpha_0$ is non-trivial and all other $\alpha_t$'s are 1.

\begin{remark}
The solutions to both \dom and \dks are allowed to be multi-sets, since we regard the facilities as mobile ones and it is natural for them to be co-located. We note that all of our hardness results also apply if we only allow the solution to consist of subsets instead of multi-sets (See Section~\ref{section-outlier}).
\end{remark}

\subsection{Our results}
\label{section:our:results}

\subsubsection{Dynamic ordered \texorpdfstring{$k$}{k}-median}
We first study \dom. When all entries of the weight vectors are larger than some small constant $\epsilon>0$, we present a constant-approximation on general metrics. Moreover, for 2-\dom, i.e., $T=2$, we present a constant-factor approximation algorithm with no assumptions on the weight vectors.
\begin{theorem}\label{thm:sample:dom}
(Informal; see Theorem~\ref{thm:dom}).
When $T=2$, there exists a polynomial-time constant-factor approximation algorithm for \emph{2-\dom}.
When $T\geq3$ is a constant, and the smallest entry in $\{w_t\}_{t=1}^T$ is at least some constant $\epsilon>0$, there exists a polynomial-time $O(\gamma/\epsilon)$-approximation algorithm for \emph{$T$-\dom}.
\end{theorem}

\subparagraph{Our techniques.}
The key idea in our algorithm is to design a surrogate LP relaxation to approximate the ordered objective, and embed its fractional solution in a network flow instance. We proceed to round the fractional flow to an integral flow, thus obtaining the induced integral solution to the original problem. The network is constructed based on a filtering process introduced by Charikar and Li \cite{charikar2012dependent}. When estimating the service cost of each client, we also adapt the oblivious clustering arguments by Byrka~\etal\cite{byrka2018constant}, with a slight increase in the approximation factor due to the structure of the network. One notable difficulty we manage to overcome when translating an integral flow into an integral solution, is that the flow oftentimes indicates the opening of more than $k$ facilities, and we need to remove some of them without incurring a cost that is unbounded compared to the LP objective.


\subsubsection{Dynamic \texorpdfstring{$k$}{k}-supplier}

We first obtain inapproximability results for \dks and its outlier version \rdks.
These two problems turn out to be much harder when $T\geq3$ than the case of $T=2$. 

\begin{theorem}
	One has the following hardness results.
	\begin{itemize}
		\item Let $T\geq3$. There are no polynomial-time algorithms for \emph{$T$-\dks} with non-trivial approximation factors, unless $\mathrm{P=NP}$ (Theorem~\ref{thm:dksnphard}).
		There exists a constant $\epsilon_0\in(0,1)$, such that \emph{$T$-\rdks} admits no multi-criteria $(\alpha,1-\epsilon_0,\dots,1-\epsilon_0)$-approximations for any non-trivial factor $\alpha$, unless $\mathrm{P=NP}$ (Theorem~\ref{rdksapxhard}).
		\item There are no polynomial-time multi-criteria $(\alpha,1,1)$-approximation algorithms for \emph{2-\rdks} for any non-trivial factor $\alpha$, unless $\mathrm{P=NP}$ (Theorem~\ref{2rdkshard}).
	\end{itemize}
\end{theorem}

On the positive side, we present a flow-based 3-approximation for 2-\dks and a matching-based multi-criteria approximation for 2-\rdks.
The approximation guarantee for 2-\dks is optimal since vanilla $k$-supplier is NP-hard to approximate within a factor of $(3-\epsilon)$ for any $\epsilon>0$ \cite{hochbaum1986unified}.
The multi-criteria approximation guarantee for 2-\rdks is also nicely complemented by the aforementioned hardness result in Theorem~\ref{2rdkshard}.

\begin{theorem}\label{thm:sample:dks}
There exists a 3-approximation for \emph{2-\dks} (Theorem~\ref{thm:dks}).
For every constant $\epsilon>0$, there exists a multi-criteria $(3,1-\epsilon,1-\epsilon)$-approximation for \emph{2-\rdks} (Theorem~\ref{thm:dksout}).
\end{theorem}

\subparagraph{Our techniques.}
Our algorithm for 2-\dks first guesses the optimum $R^\star$, and uses a standard greedy algorithm (see, e.g., \cite{gonzalez1985clustering}) to create clusters for both time steps $t=1,2$. Since the maximum distance any facility can travel is $B$, we represent all such possible movements using edges of a bipartite graph. The bipartite graph is then embedded into a network by appending the aforementioned clusters on both sides, forming a network flow instance. The link capacities in the network are all integers, thus we can directly obtain an integral flow (because the corresponding coefficient matrix is totally unimodular), which in turn induces a 3-approximation to the original problem.

For 2-\rdks, we first guess a constant-size portion of ``heavy'' facilities in the optimal solution, properly modify the instance and solve an LP relaxation on the remaining problem. This guessing step is standard in multi-objective optimization \cite{grandoni2014new}. According to the LP solution, we form clusters using the filtering algorithm by Harris~\etal\cite{harris2019lottery}, and create a bipartite matching instance where each vertex either has no contribution to the coverage of clients, or represents a cluster and provides a certain number of nearby distinct clients to cover. By assigning vertex weights, the problem of covering some specified numbers of clients for $t=1,2$ becomes finding a matching that satisfies a lower bound of total weights on both sides of the bipartite graph. We round the LP-induced fractional matching to an integral one using the iterative rounding methods by Grandoni~\etal\cite{grandoni2014new}.

\subsection{Related work}
\label{section:related:work}
The ordered $k$-median problem generalizes a number of classic clustering problems like 
$k$-center, $k$-median, $k$-facility $l$-centrum, and has been studied extensively in 
the literature. There are numerous approximation algorithms known for its special cases. 
We survey here only the results most relevant to our work (ignoring, for instance, results
regarding restricted metric spaces or fixed $k$). 
Constant approximations for 
$k$-median can be obtained via local search, Lagrangian relaxations and the primal-dual schema, 
and LP-rounding \cite{arya2004local,byrka2017improved,jain2001approximation,cohen2022improved}. 
Constant approximations for $k$-center are obtained via greedy algorithms \cite{gonzalez1985clustering,hochbaum1985best}. 
Aouad and Segev \cite{aouad2019ordered} employ the idea of surrogate models and give the first 
$O(\log n)$-approximation for ordered $k$-median. Later, Byrka~\etal\cite{byrka2018constant} and
Chakrabarty and Swamy \cite{chakrabarty2018interpolating} both successfully devise constant-factor approximations for $k$-facility $l$-centrum and ordered $k$-median. Chakrabarty and Swamy \cite{chakrabarty2019approximation} subsequently improve the approximation factor for ordered $k$-median to $(5+\epsilon)$, using deterministic rounding in a unified 
framework.

In the closely-related \emph{$k$-center with outliers} problem, a.k.a. robust $k$-center, we are required to select $k$ open facilities $F\subseteq X$ in the finite metric space $(X,d)$, $m$ served clients $S\subseteq X$, and the goal is to minimize $\max_{j\in S}d(j,F)$. This problem is introduced by Charikar~\etal in \cite{charikar2001algorithms}, where they give a greedy algorithm that achieves an approximation factor of 3. A best-possible 2-approximation is given independently by Chakrabarty~\etal\cite{chakrabarty2016non} and Harris~\etal\cite{harris2019lottery}. Many of its variants are also studied. In matroid center with outliers, the input is the same as $k$-center with outliers, except that the cardinality constraint is replaced with a given matroid, and the set of open facilities is required to be an independent set of the matroid. Chen~\etal\cite{chen2016matroid} give the first 7-approximation for this problem, and a tight 3-approximation is later obtained by Harris~\etal\cite{harris2019lottery}. In knapsack center with outliers, the input is the same as $k$-center with outliers with the cardinality constraint removed, every facility has a non-negative weight, and the total weight of open facilities is required to be no more than a given threshold. Chen~\etal\cite{chen2016matroid} give a 3-approximation that violates the knapsack constraint by a factor of $(1+\epsilon)$ for this  problem, and Chakrabarty and Negahbani \cite{chakrabarty2019generalized} give the first pure 3-approximation.

Recently, Bandyapadhyay~\etal\cite{bandyapadhyay2019constant} introduce the \emph{$T$-colorful $k$-center} problem ($T$\ckc), as a generalization of $k$-center with outliers. In this problem, given a finite metric space $(X,d)$, there are $T=O(1)$ subsets $X_t\subseteq X$ and $T$ thresholds $m_t\in\mathbb{Z}_+$ for $t\in[T]$. We are then asked to select $k$ open facilities $F\subseteq X$, $T$ served \emph{client subsets} $S_t\subseteq X_t$ satisfying $|S_t|\geq m_t$ for $t\in[T]$, and the goal is to minimize $\max_{t\in[T]}\max_{j\in S_t}d(j,F)$. Evidently, $T$\ckc recovers $k$-center with outliers by setting $T=1$. Bandyapadhyay~\etal\cite{bandyapadhyay2019constant} give a pseudo 2-approximation for $T$\ckc by opening at most $k+T-1$ facilities. Anegg~\etal\cite{anegg2020technique} and Jia~\etal\cite{jia2020fair} independently obtain pure constant factor approximations for $T$\ckc. 
It is fairly easy to see that $T$-\rdks exactly recovers $T$\ckc by setting the movement constraint $B=0$, regarding the $T$ client sets at different time steps as having $T$ colors, and removing redundant co-located open facilities from the solution. Unfortunately, unlike $T$\ckc which has pure constant approximations for $T=O(1)$, our formulation of $T$-\rdks is seemingly much harder. In many cases, pure approximations with any non-trivial factors, or even multi-criteria approximations are impossible unless $\mathrm{P=NP}$.

Our problems are closely related to the mobile facility location problems (\mfl), introduced by Demaine~\etal\cite{demaine2009minimizing}.
In these problems, a static set of clients has to be served by a set of facilities that are given initial locations and can be moved to improve the service cost at the expense of incurring a facility movement cost. For the minimum total movement \mfl problem (\tmfl), Friggstad and Salavatipour \cite{friggstad2011minimizing} give an 8-approximation using LP-rounding, where all facilities have unit weights. Ahmadian~\etal\cite{ahmadian2013local} give a local search algorithm for \tmfl with weighted facilities and proportional movement costs via $p$-swaps with an approximation factor of $3+O(\sqrt{\log\log p/\log p})$, and specifically show that the factor is at most $499$ for $p=1$. Swamy \cite{swamy2016improved} obtains an 8-approximation for the case of arbitrary movement costs using the reduction to the matroid median problem. Krishnaswamy~\etal\cite{krishnaswamy2018constant} later improve the approximation factor of matroid median to 7.081.

The dynamic formulations of our problems are closely related to the facility location problem with evolving metrics, proposed by Eisenstat~\etal\cite{eisenstat2014facility}. In this problem, there are also $T$ time steps. While the facilities and clients are fixed, the underlying metric is changing. The total cost is the sum of facility-opening cost, client-serving cost and additional switching costs for each client. The switching cost is paid whenever a client switches facility between adjacent time steps. In comparison, our problem \dks considers the cost of moving facilities instead of opening costs, and allows the number of clients to change over time. Eisenstat~\etal\cite{eisenstat2014facility} consider the problem when the open facility set $A$ is fixed, and give a $O(\log(nT))$-approximation, where $n$ is the number of clients. They also show a hardness result on $o(\log T)$-approximations. An~\etal\cite{an2017dynamic} consider the case when the open facilities are allowed to evolve as well, and give a 14-approximation.

Our problems are also related to stochastic $k$-server \cite{dehghani2017stochastic} and the page migration problem \cite{black1989competitive,westbrook1994randomized}. 
As mentioned before, the expected objective of stochastic $k$-server is the same as in \dom, if we consider non-ordered and weighted clients and all weights sum up to 1 for each time slot. Dehghani~\etal\cite{dehghani2017stochastic} provide an $O(\log n)$-approximation for stochastic $k$-server in general metrics, where $n$ is the size of the distribution support.
Our result in Theorem~\ref{thm:sample:dom} does not imply a constant approximation for stochastic $k$-server. 
The difficulty is that if one maps the stochastic $k$-server problem to ours, the corresponding weight coefficient $\gamma$ is not necessarily a constant and our approximation factor is linear in $\gamma$. 
Obtaining a constant-factor approximation algorithm for stochastic $k$-server is still an interesting open problem.

Another related dynamic formulation of clustering problems is the fully dynamic model \cite{chan2018dynamic}.
In these dynamic clustering problems, there is an adversary with a \emph{hidden} sequence of operations.
At each time step $t$, the adversary inserts a point $x_t$ or deletes a point $x_t'$, based on the $t$-th operation in its sequence.
We are required to maintain a good approximate solution after each adversarial operation, such that small running time and space are achieved in the long run.
For classic $k$-clustering objectives such as $k$-center, $k$-median and $k$-means, constant-factor approximation algorithms under fully dynamic models are developed in \cite{chan2018dynamic,cohen2019dynamic}.

\subsection{Organization}

The remainder of this paper is organized as follows. 
In Section~\ref{section-median}, we give a polynomial-time approximation algorithm for $T$-\dom based on LP rounding and a network flow instance. 
In Section~\ref{section-outlier}, we first prove the hardness of approximation for $T$-\dks and $T$-\rdks when $T\geq3$, then provide a 3-approximation for 2-\dks; we show another hardness result on pure approximations for 2-\rdks, and complete the results with a multi-criteria $(3,1-\epsilon,1-\epsilon)$-approximation for 2-\rdks.
Finally, we list some future directions and open problems in Section~\ref{section:future}.


\section{A constant approximation for \dom}\label{section-median}

In this section, we devise an LP-based algorithm for \dom. The ordered objective is estimated using reduced cost functions as introduced in~\cite{byrka2018constant} 
(see Section \ref{section:lemma:proof} for the formal definitions).
At the center of our algorithm, we construct a network flow instance by applying a modified version of the filtering algorithm by Charikar and Li~\cite{charikar2012dependent}, and use an integral flow to induce the output solution. 
When $T\geq 3$, the integral flow may open more than $k$ facilities in each time step, thus we apply a crucial subroutine in Section \ref{section:network} to reroute some part of the flow and prune the extra open facilities.
We analyze the approximation factor by adapting the oblivious rounding analysis by Byrka~\etal\cite{byrka2018constant} in Section \ref{section:median:analysis}, and provide the missing proofs in Section \ref{section:lemma:proof}.

\subsection{LP relaxation}\label{section:relaxation}

We first give the LP relaxation. By adding a superscript to every variable to indicate the time step, denote $x_{ij}^\pt\in[0,1]$ the extent of connection between client $j$ and facility $i$, and $y_i^\pt\geq0$ the extent of opening facility location $i$ at time step $t$. Moreover, denote $z_{ii'}^\pt$ the fractional movement from facility $i$ to facility $i'$ between neighboring time steps $t$ and $t+1$. We use the cost reduction trick by Byrka~\etal\cite{byrka2018constant}. Call $d':X\times X\rightarrow R_{\geq0}$ a \emph{reduced cost function} (not necessarily a metric) of metric $d$, if for any $x,y\in X$, one has $d'(x,y)\geq0$, $d'(x,y)=d'(y,x)$, and $d(x_1,y_1)\leq d(x_2,y_2)\Rightarrow d'(x_1,y_1)\leq d'(x_2,y_2)$.
For a sequence of reduced cost functions $\mathfrak{D}=\{d^\pt\}_{t=1}^T$ of $d$, the relaxation is defined as follows.
\begin{alignat*}{3}
\text{min}\quad&&\sum_{t=1}^T\sum_{j\in C_t}\sum_{i\in F_t}d^\pt(i,j)x_{ij}^\pt&+\gamma\sum_{t=1}^{T-1}\sum_{i\in F_t}&&\sum_{i'\in F_{t+1}}d(i,i')z_{ii'}^\pt\nonumber\tag{$\mathrm{LP}(\mathfrak{D})$}\label{lp-median}\\
\text{s.t.}\quad&&\sum_{i\in F_t}x_{ij}^\pt&=1&&\forall j\in C_t,t\in[T]\\
&&\sum_{i\in F_t}y_i^\pt&=k&&\forall t\in[T]\\
&&0\leq x_{ij}^\pt&\leq y_i^\pt&&\forall i\in F_t,j\in C_t,t\in[T]\\
&&\sum_{i'\in F_{t+1}}z_{ii'}^\pt&=y_i^\pt&&\forall i\in F_t,t\in[T-1]\\
&&\sum_{i\in F_{t}}z_{ii'}^\pt&=y_{i'}^{(t+1)}.&&\forall i'\in F_{t+1},t\in[T-1]
\end{alignat*}

We solve \ref{lp-median} and obtain an optimal solution $(x,y,z)$, assuming that whenever $x_{ij}^\pt>0$, we have $x_{ij}^\pt=y_i^\pt$, via the standard duplication technique on facility locations (for example, see~\cite{charikar2012dependent}). Strictly speaking, whenever we split $i\in F_t$ into co-located copies, we also need to split the corresponding variables in $y^\pt$ and $z^{(t-1)},z^\pt$ in order for it to remain feasible to the LP relaxation. Here, we first split $y^\pt$ s.t. $x_{ij}^\pt\in\{0,y_i^\pt\}$, then arbitrarily split the related $z$ variables such that the last two constraints in \ref{lp-median} are still satisfied. Denote $\oball(j,R)=\{x\in X: d(x,j)<R\}$ the open ball centered at $j$ with radius $R$, and $E_j^\pt=\{i\in F_t:x_{ij}^\pt>0\}$ the relevant facilities for client $j$. Denote $\dav^\pt(j)=\sum_{i\in F_t}d(i,j)x_{ij}^\pt$ the average service cost of client $j$ and $y^\pt(S)=\sum_{i\in S}y_i^\pt$ the ``volume'' of fractional facilities in $S\subseteq F_t$, with respect to $y^\pt$. We perform a filter-and-match algorithm (see Algorithm \ref{algo:med:filter}) to obtain a subset $C_t'\subseteq C_t$ for each $t$, a so-called ``bundle'' $\calU_j^\pt\subseteq F_t$ for each $j\in C_t'$ and a partition $P_t$ of $C_t'$, where
\begin{itemize}
    \item $C_t'$ is a subset of ``well-separated" clients of $C_t$, and for any client $j'\in C_t\setminus C_t'$, there exists a relatively close client in $C_t'$, acting like a proxy for $j'$;
    \item $\calU_j^\pt$ is a subset of fractionally open facility locations that are relatively close to client $j\in C_t'$, and the bundles in $\{\calU_j^\pt:j\in C_t'\}$ are pair-wise disjoint;
    \item $P_t$ is a judiciously created partition of $C_t'$, where every subset is either a pair of clients (a.k.a. a \emph{normal pair}), or a single client (a.k.a. a \emph{singleton pair}). Each normal pair $\{j,j'\}$ in $P_t$ is chosen such that either $j$ or $j'$ is the nearest neighbor of the other in $C_t'$, and at least one facility is opened in $\calU_j^\pt\cup\calU_{j'}^\pt$.
\end{itemize}

\begin{algorithm}[hbt!]
\caption{FILTER\&MATCH}\label{algo:med:filter}
\DontPrintSemicolon
\SetKwInOut{Input}{Input}
\SetKwInOut{Output}{Output}

\Input{$(X,d),(x,y,z),\{C_t\}_{t=1}^T,\{F_t\}_{t=1}^T$}
\Output{filtered client subsets with a bundle for each filtered client and a partition on each subset}

\For{$t\in[T]$}{
$C_t'\leftarrow\emptyset,C_t''\leftarrow C_t$\;
\While(\tcp*[f]{filtering phase}){$C_t''$ is nonempty}{
choose $j\in C_t''$ s.t. $\dav^\pt(j)$ is minimized\;
$C_t'\leftarrow C_t'\cup\{j\},C_t''\leftarrow C_t''\setminus\{j\}$, delete each $j'\in C_t''$ s.t. $d(j,j')\leq 4\dav^\pt(j')$\;}
\For{$j\in C_t'$}{
$n_j^\pt\leftarrow \argmin_{j'\in C_t',j\neq j'}d(j,j')$, $R_j^\pt\leftarrow \frac{1}{2}d(j,n_j^\pt)$, $\calU_j^\pt\leftarrow E_j^\pt\cap \oball(j,R_j^\pt)$\;}
$P_t\leftarrow\emptyset, C_t''\leftarrow C_t'$\;
\While(\tcp*[f]{matching phase}){$\exists j\in C_t''$ s.t. $n_j^\pt\in C_t''$}{
choose such $j\in C_t''$ s.t. $d(j,n_j^\pt)$ is minimized\;
$P_t\leftarrow P_t\cup\{(j,n_j^\pt)\}$, $C_t''\leftarrow C_t''\setminus\{j,n_j^\pt\}$\;}
\For{$j\in C_t''$}{
$P_t\leftarrow P_t\cup\{(j)\}$, $C_t''\leftarrow C_t''\setminus\{j\}$\;}}
\Return $\{C_t'\}_{t=1}^T,\,\{\calU_j^\pt:j\in C_t'\}_{t=1}^T,\,\{P_t\}_{t=1}^T$\;
\end{algorithm}

We first provide some basic properties of $C_t'$ and $\calU_j^\pt$ (also see the results in~\cite{charikar2012dependent}).
\begin{lemma}\label{lemma:med:filter}(Charikar and Li~\cite{charikar2012dependent}). Fix $t\in[T]$. The following statements hold. 
\begin{enumerate}[(1)]
    \item For any $j,j'\in C_t',\,j\neq j'$, one has $d(j,j')>4\max\{\dav^\pt(j),\dav^\pt(j')\}$.
    \item For any $j'\in C_t\setminus C_t'$, there exists $j\in C_t'$ s.t. $\dav^\pt(j)\leq\dav^\pt(j'),\,d(j,j')\leq4\dav^\pt(j')$.
    \item For any $j\in C_t'$, $1/2\leq y^\pt(\calU_j^\pt)\leq1$.
    \item For any $j,j'\in C_t',\,j\neq j'$, $\calU_j^\pt\cap\calU_{j'}^\pt=\emptyset$.
\end{enumerate}
\end{lemma}
\begin{proof}
The first two assertions are obvious given the order we filter the clients in Algorithm \ref{algo:med:filter}. For the third one, we have $y^\pt(\calU_j^\pt)\leq y^\pt(E_j^\pt)=1$, since we assume $x_{ij}^\pt=y_i^\pt$ whenever $i\in E_j^\pt$. For the other inequality, assume otherwise and we obtain $y^\pt(E_j^\pt\setminus\oball(j,R_j^\pt))=1-y^\pt(\calU_j^\pt)>1/2$. One also has $R_j^\pt=0.5d(j,n_j^\pt)>2\max\{\dav^\pt(j),\dav^\pt(n_j^\pt)\}\geq2\dav^\pt(j)$ because $j$ and $n_j^\pt$ are both in $C_t'$. This puts the average service cost of $j$ at least
\[\dav^\pt(j)\geq\sum_{i\in E_j^\pt\setminus\oball(j,R_j^\pt)}y_i^\pt d(i,j)\geq R_j^\pt\cdot  y^\pt(E_j^\pt\setminus\oball(j,R_j^\pt))>\dav^\pt(j),\]
which is a contradiction. For the last proposition, we simply notice that $\calU_j^\pt\subseteq\oball(j,R_j^\pt)$ and $\calU_{j'}^\pt\subseteq\oball(j',R_{j'}^\pt)$. Since the sum of two radii is at most $R_j^\pt+R_{j'}^\pt=0.5(d(j,n_j^\pt)+d(j',n_{j'}^\pt))\leq d(j,j')$, the two open balls must be disjoint.
\end{proof}

On the partition $P_t$, we discuss the differences between our Algorithm \ref{algo:med:filter} and the algorithm used by Charikar and Li~\cite{charikar2012dependent}.
In each time step $t$, Charikar and Li use a simple greedy algorithm on the filtered client set $C_t'$, matching the closest unmatched pair in $C_t'$ whenever possible, thus leaving at most one client in $C_t'$ unmatched. 
In our algorithm, we first compute for every $j\in C_t'$ its nearest neighbor $n_j^\pt\in C_t'$; 
whenever there exists an unmatched $j$ such that $n_j^\pt$ is also unmatched, we choose $j$ which minimizes $d(j,n_j^\pt)$ among such choices and match the two. 
Notice that this process may leave an arbitrary number of clients in $C_t'$ unmatched and they all end up in singleton pairs. 

The motivation in our algorithm to leave many clients unmatched, is that by restricting them to singleton pairs, the additional rerouting cost incurred in the post-processing phase (see the next section) is much easier to bound. In fact, we are not certain whether the original matching in~\cite{charikar2012dependent} can result in a good approximate solution under our rounding framework. It is also worth noting that, while we define the objective of \ref{lp-median} using reduced cost functions in $\mathfrak{D}$ that simulate the ordered service cost, Algorithm \ref{algo:med:filter} is completely oblivious of them and only uses the underlying metric $d$. 

\subsection{Flow-based LP rounding}
\label{section:network}

We construct an instance of network flow $\network$, and embed the LP solution as a fractional flow $\fracflow$. The network $\network$ consists of a source $\source$, a sink $\sink$ and $6T$ intermediate layers $L_1,L_2,\ldots,L_{6T}$ arranged in a linear fashion. For each time step $t\in[T]$, we create \emph{two} nodes for every pair $p\in P_t$, every bundle $\calU_j^\pt$ and every candidate facility location $i\in F_t$. All these nodes are contained in the layers $L_{6t-5},\ldots,L_{6t}$. To distinguish between the two mirror nodes, we use $\lnode(\cdot)$ and $\rnode(\cdot)$ to represent the nodes in $\{L_{6t-5},L_{6t-4},L_{6t-3}\}$ (on the left) and the nodes in $\{L_{6t-2},L_{6t-1},L_{6t}\}$ (on the right), respectively. The network is constructed as follows, and an illustration is given in Fig. \ref{figure:network}.
 
\begin{enumerate}[Step 1.]
\item For each $i\in F_t,\,t\in[T]$, add nodes $\lnode(i)$ to $L_{6t-5}$ and $\rnode(i)$ to $L_{6t}$.
\item For each $\calU_j^\pt,\,t\in[T]$, add nodes $\lnode(\calU_j^\pt)$ to $L_{6t-4}$ and $\rnode(\calU_j^\pt)$ to $L_{6t-1}$.
\item For each $p\in P_t,\,t\in[T]$, add nodes $\lnode(p)$ to $L_{6t-3}$ and $\rnode(p)$ to $L_{6t-2}$.
\item For each $j\in C_t',\,p\in P_t$ s.t. $j\in p,\,t\in[T]$, add links $(\lnode(\calU_j^\pt),\lnode(p)),(\rnode(p),\rnode(\calU_j^\pt))$ both with capacity range $\left[\floor{ y^\pt(\calU_{j}^\pt)},\ceil{y^\pt(\calU_{j}^\pt)}\right]$ and initial fractional flow values $\fracflow(\lnode(\calU_j^\pt),\lnode(p))=\fracflow(\rnode(p),\rnode(\calU_j^\pt))=y^\pt(\calU_{j}^\pt)$. Note that the capacity is either $[0,1]$ or $\{1\}$.
\item For each $p\in P_t,\,t\in[T]$, add the link $(\lnode(p),\rnode(p))$ with capacity $[\floor{y^\pt(p)},\ceil{y^\pt(p)}]$ and define its flow as $\fracflow(\lnode(p),\rnode(p))=y^\pt(p)=\sum_{j\in p}y^\pt(\calU_j^\pt)$. 
According to Lemma \ref{lemma:med:filter}, if $p$ is a normal pair, the capacity is either $[1,2]$ or $\{1\}$ or $\{2\}$, since $y^\pt(p)\geq2\times(1/2)=1$; if $p$ is a singleton pair, the capacity is either $[0,1]$ or $\{1\}$.
\item For each $j\in C_t'$ and $i\in \calU_j^\pt,\,t\in[T]$, add links $(\lnode(i),\lnode(\calU_j^\pt)),(\rnode(\calU_j^\pt),\rnode(i))$ with unit capacity $[0,1]$. Let the initial fractional flows be $\fracflow(\lnode(i),\lnode(\calU_j^\pt))=\fracflow(\rnode(\calU_j^\pt),\rnode(i))=y_i^\pt$.
\item For each $i\in F_t\setminus(\bigcup_{j\in C_t'}\calU_j^\pt),\,t\in[T]$, add the link $(\lnode(i),\rnode(i))$ with capacity $[\floor{y_i^\pt},\ceil{y_i^\pt}]$ (\emph{across intermediate layers $L_{6t-4},\ldots,L_{6t-1}$}). Let its initial fractional flow be $\fracflow(\lnode(i),\rnode(i))=y_i^\pt$. Note that the flow may be larger than one, since we allow multi-sets and do \emph{not} impose the constraint $y_i^\pt\leq1$ in \ref{lp-median}.
\item For each $i\in F_t,i'\in F_{t+1},\,t\in[T-1]$, add the link $(\rnode(i),\lnode(i'))$ with capacity $[\floor{z_{ii'}^\pt},\ceil{z_{ii'}^\pt}]$. Let its initial fractional flow be $\fracflow(\rnode(i),\lnode(i'))=z_{ii'}^\pt$.
\item For each $i\in F_1$, add the link $(\source,\lnode(i))$ with capacity $[\floor{y_i^\pone},\ceil{y_i^\pone}]$ and initial flow $\fracflow(\source,\lnode(i))=y_i^\pone$. For each $i'\in F_T$, add the link $(\rnode(i'),\sink)$ with capacity $[\floor{y_{i'}^{(T)}},\ceil{y_{i'}^{(T)}}]$ and initial flow $\fracflow(\rnode(i'),\sink)=y_{i'}^{(T)}$.
\end{enumerate}

\begin{figure}[hbt!]
    \centering
    \includegraphics{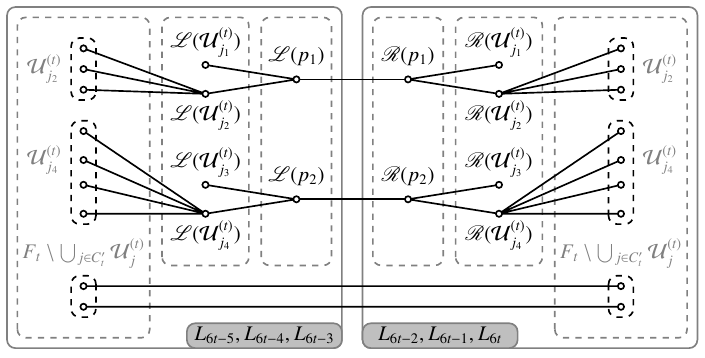}
    \caption{Several intermediate layers of $\network$ representing a single time step $t$. Some nodes and links are left out for simplicity.}
    \label{figure:network}
\end{figure}

Notice $\fracflow$ is naturally a flow with value $k$, since the flow conservation constraints are directly satisfied by the last two constraints of \ref{lp-median}. 
Because the flow polytope is defined by a totally unimodular matrix, and our capacity constraints are all integers, it is well-known (see, e.g., \cite{schrijver2003combinatorial}) that the flow polytope has integral extreme point solutions.
Using the dependent rounding algorithm by Kumar~\etal\cite{kumar2009unified}, one can efficiently sample an integral flow $\intflow$ corresponding to an integral extreme point solution,
such that $\intflow$ is guaranteed to have value $k$, and $\E[\intflow]=\fracflow$ holds for all links (cf. \cite{hajiaghayi2016constant}).
Next, given the integral flow $\intflow$, we deterministically construct the solution $\{A_t\}_{t=1}^T$ as follows,
\begin{itemize}
	\item If $T=2$, there are 12 layers $L_1,L_2,\ldots,L_{12}$ in the network. For each link $e=(\rnode(i_1),\lnode(i_2))$ between $L_6$ and $L_7$ such that $\intflow(e)=m\geq1$, add $m$ copies of $i_1$ to $A_1$ and $m$ copies of $i_2$ to $A_2$.
	\item If $T\geq 3$, $\intflow$ does not immediately reveal a feasible solution. To see this, focus on any unit flow in $\intflow$. It may enter $L_7$ and exit from $L_{12}$ (both for $t=2$) through nodes that represent different facility locations. We design the following Algorithm \ref{algo:reroute} to resolve this issue.
	 
	For a facility $i\in F_t$, if there is (at least) one unit of flow through $\lnode(i)$ or $\rnode(i)$, we call the facility $i$ \emph{left-activated} or \emph{right-activated} correspondingly. The algorithm looks at each pair $(j_1,j_2)=p\in P_t$ independently, and considers the 1 or 2 units of flow $\intflow$ on the link $(\lnode(p),\rnode(p))$. There are two cases: The first one is when $\intflow(\lnode(p),\rnode(p))=2$. Since the links $(\lnode(\calU_{j_1}^\pt),\lnode(p))$ and $(\lnode(\calU_{j_2}^\pt),\lnode(p))$ both have capacities at most one, one must also have $\intflow(\lnode(\calU_{j_1}^\pt),\lnode(p))=\intflow(\lnode(\calU_{j_2}^\pt),\lnode(p))=1$. This happens for the nodes on the right (i.e., in layers $L_{6t-2},L_{6t-1},L_{6t}$) as well. In this case, we simply ignore the activated facilities in $L_{6t}$, and add the activated facilities in $L_{6t-5}$ to $A_t$; the other case is when $\intflow(\lnode(p),\rnode(p))=1$. Now, the two activated facilities in $L_{6t-5}$ and $L_{6t}$ may not even be in the same bundle. Here, we exploit the crucial property of $p$ that either $j_1$ or $j_2$ is the nearest neighbor of the other in $C_t'$, and are able to obtain a simple criterion of choosing whether the left-activated or right-activated facility location. The cases of singleton pairs and facility locations in $F_t\setminus\bigcup_{j\in C_t'}\calU_j^\pt$ are similar and much easier to handle, since the integral flow is less ambiguous for them. 
\end{itemize}

\begin{algorithm}[hbt!]
\caption{REROUTE}\label{algo:reroute}
\DontPrintSemicolon
\SetKwInOut{Input}{Input}
\SetKwInOut{Output}{Output}

\Input{$\network,\intflow,\{C_t'\}_{t=1}^T,\,\{\calU_j^\pt:j\in C_t'\}_{t=1}^T,\,\{P_t\}_{t=1}^T$}
\Output{a feasible solution to the original $T$-\dom instance}
\For{$t\in[T]$}{
$A_t\leftarrow\emptyset$\;
\For{$p\in P_t$}{
\uIf{$\intflow(\lnode(p),\rnode(p))=2$}{
pick the left-activated $i_1,i_2$, $A_t\leftarrow A_t\cup\{i_1,i_2\}$\;}
\ElseIf{$\intflow(\lnode(p),\rnode(p))=1$}{
\uIf(\tcp*[f]{$j_1$ and $j_2$ are closest to each other}){$p=(j_1,j_2)$, $n_{j_1}^\pt=j_2,n_{j_2}^\pt=j_1$}{
pick the left-activated $i$, $A_t\leftarrow A_t\cup\{i\}$\;}
\uElseIf{$p=(j_1,j_2)$, $n_{j_1}^\pt=j_2,n_{j_2}^\pt\neq j_1$}{
\uIf{$\intflow$ passes through \emph{the same bundle} in $L_{6t-5}$ and $L_{6t}$}{
pick the left-activated $i$, $A_t\leftarrow A_t\cup\{i\}$\;}
\Else{
pick the (left or right) activated $i\in\calU_{j_2}^\pt$, $A_t\leftarrow A_t\cup\{i\}$\;}}
\ElseIf{$p=(j)$}{
pick the left-activated $i$, $A_t\leftarrow A_t\cup\{i\}$\;}
}}
\For{$i\in F_t\setminus\bigcup_{j\in C_t'}\calU_j^\pt$}{
\If{$\intflow(\lnode(i),\rnode(i))=m\geq1$}{
add $m$ copies of $i$ to $A_t$\;}}}
\Return $\{A_t\}_{t=1}^T$\;
\end{algorithm}

\begin{figure}[hbt!]
    \centering
    \includegraphics{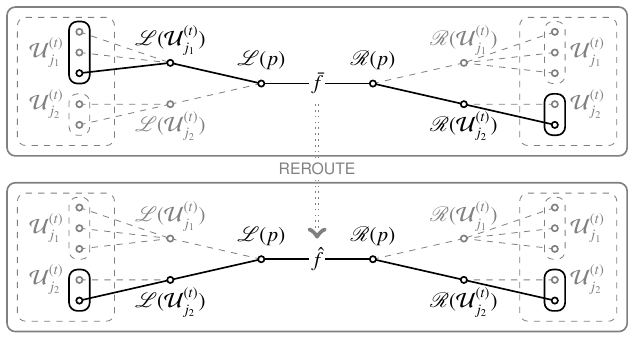}
    \caption{An illustration of Algorithm \ref{algo:reroute}, where $p=(j_1,j_2),\,n_{j_1}^\pt=j_2$ and $n_{j_2}^\pt\neq j_1$. We highlight the unit flow on $(\lnode(p),\rnode(p))$ before and after rerouting using solid links, and use dashed lines to represent links with no flow on them. The algorithm resolves the ambiguity in $\intflow$ by choosing the right-activated facility location in $\calU_{j_2}^\pt$, which is equivalent to rerouting the flow through its mirror node on the left in $L_{6t-5}$.}
    \label{figure:reroute}
\end{figure}

One may also regard Algorithm \ref{algo:reroute} as implicitly rerouting some flows of $\intflow$ to make it \emph{symmetric} on the six layers $L_{6t-5},L_{6t-4},\dots,L_{6t}$, for every $t\in[T]$ (see Fig. \ref{figure:reroute}). This results in another integral flow that has marginal distributions different from $\intflow$, in particular the marginal distributions of $A_t,\,t\in[T]$. We first estimate the movement cost of the solution $\{A_t\}_{t=1}^T$. To this end, we have the following lemma.

\begin{lemma}\label{translemma}
Recall that $m(A,A')$ is the minimum total distance of perfect matching between $A$ and $A'$. If $T=2$, the expected movement cost of solution $\{A_1,A_2\}$ satisfies
\[\E[m(A_1,A_2)]\leq\sum_{i\in F_1}\sum_{i'\in F_2}d(i,i')z_{ii'}^{(1)}.\]

If $T\geq3$, the expected movement cost of solution $\{A_t\}_{t=1}^T$ satisfies
\[\E\left[\sum_{t\in[T-1]}m(A_t,A_{t+1})\right]\leq\sum_{t\in[T-1]}\sum_{i\in F_t}\sum_{i'\in F_{t+1}}d(i,i')z_{ii'}^\pt+6\sum_{t\in[T]}\sum_{j\in C_t}\dav^\pt(j).\]
\end{lemma}

\begin{proof}
If $T=2$, recall that we choose $A_1,A_2$ solely based on $\intflow$ and the links between $L_6$ and $L_7$, thus the matching cost of $A_1,A_2$ is obviously at most $\sum_{i\in F_1,i'\in F_2}\intflow(\rnode(i),\lnode(i'))d(i,i')$. Because $\E[\intflow]=\fracflow$, each flow on link $e=(\rnode(i),\lnode(i'))$ between $L_6$ and $L_7$ has expectation $z_{ii'}^\pone$. Hence the expectation of total movement cost is at most,
	\[\E[m(A_1,A_2)]\leq\sum_{i\in F_1,i'\in F_2}\E[\intflow(\rnode(i),\lnode(i'))]\cdot d(i,i')=\sum_{i\in F_1}\sum_{i'\in F_2}d(i,i')z_{ii'}^{(1)}.\]
	
Now we consider the case where $T\geq 3$. Denote $\intflow(L)$ the multi-set of facilities that $\intflow$ activates in layer $L$ (where $L$ has an index of $6t-5$ or $6t$, $t\in[T]$). 
According to Algorithm \ref{algo:reroute}, $|\intflow(L)|=k$ for each $L\in\{L_{6t-5},L_{6t}\}$, and we define $m_{P_t}(\intflow(L_{6t-5}),\intflow(L_{6t}))$ the total distance of a perfect matching between $\intflow(L_{6t-5})$ and $\intflow(L_{6t})$, \emph{restricted to} the matchings within each pair $p\in P_t$ (note that this does not affect $F_t\setminus\bigcup_{j\in C_t'}\calU_j^\pt$, because $\intflow$ is naturally consistent for them). 
Using the triangle inequality, one has
\begin{align}
\E\left[\sum_{t\in[T-1]}m(A_t,A_{t+1})\right]&\leq
\E\left[\sum_{t\in[T-1]}m(A_t,\intflow(L_{6t}))+m(\intflow(L_{6t}),\intflow(L_{6t+1}))+m(\intflow(L_{6t+1}),A_{t+1})
\right]\notag\\
&\leq\sum_{t\in[T]}\E\left[m(A_t,\intflow(L_{6t}))+m(\intflow(L_{6t-5}),A_t)\right]+\sum_{t\in[T-1]}\E\left[m(\intflow(L_{6t}),\intflow(L_{6t+1}))\right]\notag\\
&\leq\sum_{t\in[T]}\E\left[m(A_t,\intflow(L_{6t}))+m(\intflow(L_{6t-5}),A_t)\right]+\sum_{t\in[T-1]}\sum_{i\in F_t}\sum_{i'\in F_{t+1}}d(i,i')z_{ii'}^\pt\notag\\
&\leq\sum_{t\in[T]}\E\left[m_{P_t}(\intflow(L_{6t-5}),\intflow(L_{6t}))\right]+\sum_{t\in[T-1]}\sum_{i\in F_t}\sum_{i'\in F_{t+1}}d(i,i')z_{ii'}^\pt,\label{eq:triangle:inequality:movement}
\end{align}
where the second to last inequality follows from the same analysis as the case when $T=2$. To see the last inequality, notice that by definition of Algorithm \ref{algo:reroute}, we have $A_t\subseteq\intflow(L_{6t-5})\cup\intflow(L_{6t})$ and \[m(A_t,\intflow(L_{6t}))+m(\intflow(L_{6t-5}),A_t)\leq m_{P_t}(A_t,\intflow(L_{6t}))+m_{P_t}(\intflow(L_{6t-5}),A_t),\]
when the matchings are also restricted within pairs of $P_t$. Then consider each pair $p\in P_t$ separately. When $\intflow(\lnode(p),\rnode(p))=2$, $A_t$ always agrees with the left-activated facilities in $\intflow(L_{6t-5})$; when $\intflow(\lnode(p),\rnode(p))=1$, $A_t$ must agree with an activated facility either in $\intflow(L_{6t-5})$ or in $\intflow(L_{6t})$. Summing over all pairs, one has the following equality, which verifies~\eqref{eq:triangle:inequality:movement},
\[m_{P_t}(A_t,\intflow(L_{6t}))+m_{P_t}(\intflow(L_{6t-5}),A_t)=m_{P_t}(\intflow(L_{6t-5}),\intflow(L_{6t})).\]
 
In what follows, we fix $t$ and leave out the superscript for convenience. Denote the random variable $\Delta_t= m_{P_t}(\intflow(L_{6t-5}),\intflow(L_{6t}))$ and $\Delta_{t,p}$ the partial matching cost within pair $p\in P_t$, thus $\Delta_t=\sum_{p\in P_t}\Delta_{t,p}$. To further obtain an upper bound on $\Delta_{t,p}$, we prioritize the facilities in the same bundle, and only make a cross-bundle matching if we have to. For example for $p=(j_1,j_2)$, if $\intflow(\lnode(p),\rnode(p))=2$, with left-activated facilities $i_1\in\calU_{j_1},i_2\in\calU_{j_2}$ and right-activated ones $i_1'\in\calU_{j_1},i_2'\in\calU_{j_2}$, we (possibly suboptimally) match $(i_1,i_1')$ and $(i_2,i_2')$. We only consider cross-bundle pairs when the unit flow on the link $(\lnode(p),\rnode(p))$ passes through different bundles in $L_{6t-5}$ and $L_{6t}$, in which case we have $d(i_1,i_2')\leq d(i_1,j_1)+d(j_1,j_2)+d(j_2,i_2')$ and pay the cost $d(j_1,j_2)$. It is then easy to obtain the following using triangle inequality (we omit the singleton case here because it is obviously easier),
\begin{align}
\E[\Delta_{t,p}]\leq&\sum_{i_1,i_1'\in\calU_{j_1}}\Pr[\intflow\textrm{ left-activates }i_1,\textrm{ right-activates } i_1']d(i_1,i_1')\notag\\
&+\sum_{i_2,i_2'\in\calU_{j_2}}\Pr[\intflow\textrm{ left-activates }i_2,\textrm{ right-activates } i_2']d(i_2,i_2')\notag\\
&+\sum_{i_1\in \calU_{j_1},i_2\in\calU_{j_2}}\Pr[\intflow(\lnode(p),\rnode(p))=1,\intflow\textrm{ activates }i_1,i_2]d(i_1,i_2)\notag\\
\leq&\sum_{i\in\calU_{j_1}}d(i,j_1)\cdot\Pr[\intflow(\lnode(i),\lnode(\calU_{j_1}))=1]+d(i,j_1)\cdot\Pr[\intflow(\rnode(\calU_{j_1}),\rnode(i))=1]\notag\\
&+\sum_{i\in\calU_{j_2}}d(i,j_2)\cdot\Pr[\intflow(\lnode(i),\lnode(\calU_{j_2}))=1]+d(i,j_2)\cdot\Pr[\intflow(\rnode(\calU_{j_2}),\rnode(i))=1]\notag\\
&+\Pr[\intflow(\lnode(p),\rnode(p))=1,L_{6t-5},L_{6t}\textrm{ disagree on }\calU_{j_1},\calU_{j_2}]\cdot d(j_1,j_2)\notag\\
\leq&\,2\dav(j_1)+2\dav(j_2)+\Pr[\intflow(\lnode(p),\rnode(p))=1,L_{6t-5},L_{6t}\textrm{ disagree on }\calU_{j_1},\calU_{j_2}]\cdot d(j_1,j_2).\label{eq:disagree}
\end{align}
	
Since $(j_1,j_2)$ is a pair, w.l.o.g. let $j_2$ be the nearest neighbor of $j_1$ among $C_t'$, $R_{j_1}=0.5d(j_1,j_2)$ and $\dav(j_1)\geq(1-y(\calU_{j_1}))\cdot R_{j_1}=0.5(1-y(\calU_{j_1}))\cdot d(j_1,j_2)$, hence $(1-y(\calU_{j_1}))d(j_1,j_2)\leq 2\dav(j_1)$. Furthermore, we can obtain a simple bound on the probability in~\eqref{eq:disagree},
\begin{align*}
&\Pr[\intflow(\lnode(p),\rnode(p))=1,L_{6t-5},L_{6t}\textrm{ disagree on }\calU_{j_1},\calU_{j_2}]\\
=&\Pr[\intflow(\lnode(\calU_{j_1}),\lnode(p))=1,\intflow(\lnode(p),\rnode(p))=1,\intflow(\rnode(p),\rnode(\calU_{j_2}))=1]\\
&+\Pr[\intflow(\lnode(\calU_{j_2}),\lnode(p))=1,\intflow(\lnode(p),\rnode(p))=1,\intflow(\rnode(p),\rnode(\calU_{j_1}))=1]\\
\leq&\Pr[\intflow(\rnode(p),\rnode(\calU_{j_1}))=0]+\Pr[\intflow(\lnode(\calU_{j_1}),\lnode(p))=0]\\
\leq&\,2(1-y(\calU_{j_1})).
\end{align*}
	
Therefore, the expectation above can be further bounded as
\[\E[\Delta_{t,p}]\leq2\dav(j_1)+2\dav(j_2)+2d(j_1,j_2)\cdot (1-y(\calU_{j_1}))\leq 6\dav(j_1)+2\dav(j_2).\]
	
Notice that for each $t$, $P_t$ is a partition of $C_t'\subseteq C_t$. By summing over all pairs and all time steps, we have
\[\sum_{t\in[T]}\E[\Delta_t]\leq6\sum_{t\in [T]}\sum_{j\in C_t'}\dav^\pt(j)\leq 6\sum_{t\in [T]}\sum_{j\in C_t}\dav^\pt(j),\]
which yields the lemma when combined with~\eqref{eq:triangle:inequality:movement}.
\end{proof}

\subsection{Analysis}
\label{section:median:analysis}

For the solution $\{A_t:t\in[T]\}$ given by Algorithm \ref{algo:reroute}, recall that the service cost vector at time $t$ is defined as $d(C_t,A_t)=(d(j,A_t))_{j\in C_t}$. We first provide a lemma that bounds the expectation of $\topl$ norm of the service cost vector at time $t$, for any $\ell\in[|C_t|]$, where the $\topl$ norm of a non-negative vector is the sum of its largest $\ell$ entries. Consequently, the more general ordered cost can be written as a conic combination of $\topl$ norms, and easily bounded. A simpler version is presented as Lemma~4.7 in~\cite{byrka2018constant}, which is the result of dependent rounding by Charikar and Li~\cite{charikar2012dependent}. To obtain our following lemma, we need to further examine the rerouting procedures in Algorithm \ref{algo:reroute}, and consider the new marginal distributions of $A_t,t\in[T]$.

\begin{lemma}\label{budgetlemma}(Adapted from~\cite{byrka2018constant}).
Fix $t\in[T]$ and let $\ell\in[|C_t|],h>0$ be arbitrary. Define $d_{-h}(j,j')=0$ if $d(j,j')<h$ and $d_{-h}(j,j')=d(j,j')$ otherwise. One has
\[\E[\topl(d(C_t,A_t))]\leq 41.33\ell\cdot h+41.33\sum_{j\in C_t}\sum_{i\in F_t}d_{-h}(i,j)x_{ij}^\pt.\]
\end{lemma}

We prove Lemma \ref{budgetlemma} in Section \ref{section:lemma:proof}. For now, we turn to the case with general weights. We use an argument by Byrka~\etal\cite{byrka2018constant}. Suppose all pair-wise distances are distinct (via small perturbation) and the weight vector $w_t$ has $N_t$ distinct entries $\{\bar w_{tr}:r=1,\dots,N_t\}$ in decreasing order. For each distinct entry, we guess the exact value $T_r^\pt$ which is the smallest distance that is multiplied with $\bar w_{tr}$ in some fixed optimum. Also let $T_0^\pt=\infty$ and $T_{N_t+1}^\pt=0$. Now, the $t$-th reduced cost function in \ref{lp-median} is defined as $d^\pt(i,j)=d(i,j)\bar w_{ts}$, where $T_s^\pt\leq d(i,j)< T_{s-1}^\pt$. After solving the corresponding relaxation \ref{lp-median} and using Algorithm \ref{algo:reroute} to obtain the solution $\{A_t:t\in[T]\}$, we have the following lemma. The proof is in spirit similar to Lemma~5.1 in~\cite{byrka2018constant}, which we provide in Section \ref{section:lemma:proof}.

\begin{lemma}\label{corelem}
When $T=2$, $\{A_t:t\in[T]\}$ is an $82.66$-approximation for 2-\dom. If $T\geq 3$ is a constant and the smallest entry in $\{w_t\}_{t=1}^T$ is at least some constant $\epsilon>0$, $\{A_t:t\in[T]\}$ is an $(82.66+6\gamma/\epsilon)$-approximation for $T$-\dom. In both cases, the algorithm runs in $\prod_{t=1}^T\left(|F_t|\cdot|C_t|\right)^{O(N_t)}$ time, where $N_t$ is the number of distinct entries in the weight vector $w_t$, $t\in[T]$.
\end{lemma}
\begin{remark}
We remark that when $T\geq3$, our algorithm only works if the smallest entry in the weight vectors is at least $\epsilon>0$.
This means that our algorithm cannot handle 0-1 vectors.	
This technical difficulty arises because for 0-1 vectors, the contribution of a filtered client in $C_t'$ may be zero (in the inner products) and it is unclear how to bound the extra rerouting cost in Lemma \ref{translemma} in terms of the actual contributions of clients.
We leave the more general case as an interesting open question.
\end{remark}

One obvious limitation of the algorithm above is the prohibitive running time when $N_t=\omega(1)$. 
Here, we use the following logarithmic bucketing trick by Aouad and Segev~\cite{aouad2019ordered} and Byrka~\etal\cite{byrka2018constant}: 
To obtain a constant-factor approximation in polynomial time, one does not have to strictly adhere to the original weight vectors, and the guessed thresholds only need to be approximate. 
Since $T$ is a constant, it is possible to guess the largest service distance for each time step by losing a polynomial factor in the running time. 
The connection distances are grouped into logarithmically many buckets thus losing only a factor of $1+\delta$. 
For each bucket, its average weight is also guessed up to a small multiplicative error of $\delta$. 
There are at most $O\left(\log_{1+\delta}\left(\frac{n}{\delta}\right)\right)=O\left(\frac{1}{\delta}\log\left(\frac{n}{\delta}\right)\right)$ buckets for each time step, where $n=|F_t|+|C_t|$, therefore guessing a non-increasing sequence of the average weights only causes another polynomial factor $\exp\left(O\left(\frac{1}{\delta}\log\left(\frac{n}{\delta}\right)\right)\right)=(n/\delta)^{O\left(1/\delta\right)}$ in the running time. 
Finally, because $T$ is a constant, the overall number of guesses is still bounded by a polynomial $(n/\delta)^{O\left(T/\delta\right)}$. 
Formally, we have the following main theorem adapted from Theorem~6.1 in~\cite{byrka2018constant}, and provide a simplified proof in Section \ref{section:lemma:proof}.

\begin{theorem}\label{thm:dom}
When $T=2$, for any $\delta>0$ there exists an $82.66(1+\delta)$-approximation algorithm for 2-\dom, with running time $\left(|F_1|+|C_1|\right)^{O(1/\delta)}\cdot \left(|F_2|+|C_2|\right)^{O(1/\delta)}$. When $T\geq3$ is a constant, and the smallest entry in $\{w_t\}_{t=1}^T$ is at least some constant $\epsilon>0$, for any $\delta>0$ there exists an $\left(82.66+6\gamma/\epsilon\right)(1+\delta)$-approximation algorithm for $T$-\dom, with running time $\prod_{t=1}^T\left(|F_t|+|C_t|\right)^{O(1/\delta)}$.
\end{theorem}

\subsection{Missing proofs}
\label{section:lemma:proof}

\begin{proof}[Proof of Lemma \ref{budgetlemma}]
We fix and leave out the superscript $t\in[T]$ in the following. The proof follows the one of Lemma~4.7 by Byrka~\etal in~\cite{byrka2018constant}, by splitting the service cost of $j\in C_t$ into a deterministic part $D_j$ and another stochastic part $X_j$, s.t. $\Pr[d(j,A_t)\leq D_j+X_j]=1$. One major difficulty we need to surmount is the fact that Algorithm \ref{algo:reroute} displaces some facilities, and $A_t$ no longer follows the original marginal distributions described by Charikar and Li~\cite{charikar2012dependent}. We emphasize that our modified matching phase in Algorithm \ref{algo:med:filter} and the structure of the network enable our analysis.

Like the definition of average service cost $\dav(j)=\sum_{i\in F_t}d(i,j)x_{ij}$, we define the average \emph{reduced} service cost $\dhav(j)=\sum_{i\in F_t}d_{-h}(i,j)x_{ij}$, where $d_{-h}$ is defined as in the lemma. In the following, for any fixed client $j$, we progressively charge parts of $d(j,A_t)$ to either $D_j$ or $X_j$ s.t. $d(j,A_t)\leq D_j+X_j$ always holds, where we initially charge 0 to $D_j$ and $X_j$. By Lemma \ref{lemma:med:filter}, there must exist $j'\in C_t'$ (probably $j=j'$ when $j\in C_t'$) such that $d(j,j')\leq4\dav(j)$ and $\dav(j')\leq\dav(j)$. Notice by definition of $d_{-h}$, it is obvious that $\dav(j)\leq\dhav(j)+h$. Using triangle inequality, one has $d(j,A_t)\leq d(j',A_t)+d(j,j')\leq d(j',A_t)+4\dav(j)\leq d(j',A_t)+4\dhav(j)+4h$. We charge $4h$ to $D_j$ and $4\dhav(j)$ (a fixed value) with probability 1 to $X_j$, so $D_j=4h$ and $\E[X_j]=4\dhav(j)$ at the moment.

Next, we charge the stochastic service cost $d(j',A_t)$ to $D_j$ or $X_j$, and consider the case where $j'$ and its nearest neighbor $j''\in C_t'$ are \emph{not} matched in Algorithm \ref{algo:med:filter} (the case where they are matched is simpler but somewhat different, which will be explained later). 
Fix a constant $\beta>5$ which is determined afterwards. 
We have the following cases. 

\begin{description}
\item[Case 1: $A_t\cap\oball(j',\beta h)\neq\emptyset$.] We have $d(j',A_t)\leq\beta h$ and charge $\beta h$ to $D_j$. 

\item[Case 2: $A_t\cap\left(\calU_{j'}\backslash\oball(j',\beta h)\right)\neq\emptyset$.]
We charge this (stochastic) cost to $X_j$. According to Algorithm \ref{algo:reroute}, if $j'$ is in a singleton pair, the marginal distribution over bundle $\calU_{j'}$ is not changed, hence the expectation of this stochastic cost is at most
\begin{align}
\sum_{i\in\calU_{j'}\backslash\oball(j',\beta h)}x_{ij'}d(i,j')&=\sum_{i\in\calU_{j'}\backslash \oball(j',\beta h)}x_{ij'}d_{-h}(i,j')\notag\\
&\leq \sum_{i\in E_{j'}\backslash \oball(j',\beta h)}x_{ij'}d_{-h}(i,j')\triangleq \dhfar(j')\leq\dhav(j'),\label{eq:farcost}
\end{align}
where the last inequality is due to the definition of $\dhav(j')$. 

When $j'$ is not in a singleton pair, since $j'$ is not matched with its nearest neighbor $j''$, it is matched with $\tilde j$ in $p=(j',\tilde j)$, and $j'$ is the nearest neighbor of $\tilde j$ in $C_t'$. The marginal distribution over $\calU_{j'}$ is modified in this case, because whenever the rounded $\intflow$ disagrees on the bundles $\calU_{j'}$ versus $\calU_{\tilde j}$ in layers $L_{6t-5}$ and $L_{6t}$ (note that this only happens when $\intflow(\lnode(p),\rnode(p))=1$), Algorithm \ref{algo:reroute} always chooses the activated facility in $\calU_{j'}$. Comparing to the singleton case in~\eqref{eq:farcost}, there is some extra cost corresponding to the event that $\intflow$ right-activates some facility in $\calU_{j'}\setminus\oball(j',\beta h)$ in $L_{6t}$ and left-activates some facility in $\calU_{\tilde j}$ in $L_{6t-5}$, and the expectation is obviously no more than the amount in~\eqref{eq:farcost}.
In summary, the increase of $\E[X_j]$ in this case is at most $2\dhfar(j')$.

\item[Case 3: $A_t\cap\calU_{j'}=\emptyset$.]
Because $j'$ is not matched with its nearest neighbor $j''$, $j''$ has to be matched with another $j'''$ in another pair $p'=(j'',j''')$ s.t. $d(j'',j''')\leq d(j',j'')$. Denote $R=R_{j'}=\frac{1}{2}d(j',j'')$, since our algorithm guarantees that there exists an open facility in $\calU_{j''}\cup\calU_{j'''}$ (since we have initial flow $\fracflow(\lnode(p'),\rnode(p'))\geq1$), $j'$ can always be served by such a facility at a distance
\[d(j',j'')+d(j'',j''')+\max\{R_{j''},R_{j'''}\}\leq 4R+\max\{R_{j''},R_{j'''}\}\leq 5R,\]
where the last inequality holds when $j'''$ is the nearest neighbor of $j''$, in which case $R_{j'''}\leq R_{j''}= \frac{1}{2}d(j'',j''')\leq R$, or vice versa and $R_{j''}\leq R_{j'''}=\frac{1}{2}d(j'',j''')\leq R$. 
If $R\leq\beta h$, we simply charge $5R\leq 5\beta h$ to $D_j$, hence we assume $R>\beta h$ in the following. 

When $j'$ is in a singleton pair, we have $E_{j'}\cap \oball(j',\beta h)\subseteq \calU_{j'}$ from $R>\beta h$, and 
\begin{equation}
    \dhfar(j')=\sum_{i\in E_{j'}\backslash \oball(j',\beta h)}x_{ij'}d_{-h}(i,j')\geq\sum_{i\in E_{j'}\backslash \calU_{j'}}x_{ij'}d_{-h}(i,j')\geq R\sum_{i\in E_{j'}\backslash \calU_{j'}}x_{ij'}=R(1-y(\calU_{j'}));\label{eq:farthercost}
\end{equation}
according to our rounding algorithm, there is a probability exactly $1-y(\calU_{j'})$ that none of the facilities in $\calU_{j'}$ is chosen, in which case we charge $5R$ to $X_j$, resulting in an increase of expectation of at most $5R(1-y(\calU_{j'}))\leq 5\dhfar(j')$ to $\E[X_j]$. 
When $j'$ is in a normal pair with $\tilde j$, $j'$ must be the nearest neighbor of $\tilde j$, and the marginal probability of not selecting any facility in $\calU_{j'}$ is even smaller, hence the argument above still holds.
\end{description}

Up till now, we have charged costs to $D_j$ and $X_j$ in three different cases, such that $D_j\leq(4+5\beta)h$ (the deterministic part from $d(j',A_t)$ is always at most $5\beta h$) and $\E[X_j]\leq 4\dhav(j)+7\dhfar(j')$ (by taking the sum over all cases), only for the case when $j'$ is not matched with its nearest neighbor $j''$. 

Next, we briefly discuss the case when $j'$ is actually matched with $j''$ in $p=(j',j'')$. 
First, if $j'$ is not the nearest neighbor of $j''$, Case~1 has the same analysis; 
for Case~2, since the probability that none of the facilities in $\calU_{j'}$ is chosen is in fact \emph{increased}, the analysis is similar to \eqref{eq:farcost} with an increase of $\dhfar(j')$ to $\E[X_j]$ instead of $2\dhfar(j')$;
for Case~3, we must have $\intflow(\lnode(p),\rnode(p))=1$, and the event $A_t\cap\calU_{j'}=\emptyset$ happens with probability at most $3(1-y(\calU_{j'}))$ (considering the event that $\intflow$ naturally activates only $\calU_{j''}$ in both $L_{6t-5}$ and $L_{6t}$ with prob. at most $\Pr[\intflow(\lnode(\calU_{j'}),\lnode(p))=0]=1-y(\calU_{j'})$, and the event that the disagreement happens with prob. at most $2(1-y(\calU_{j'}))$ using the same argument in the proof of Lemma \ref{translemma}). 
Since at least one facility is opened in the bundle pair $(\calU_{j'},\calU_{j''})$ at most $d(j',j'')+R_{j''}\leq 3R$ away from $j'$, we (i) either charge $3\beta h$ to $D_j$ if $R\leq\beta h$, (ii) or charge $3R$ more to $X_j$ with probability at most $3(1-y(\calU_{j'}))$ if $R>\beta h$, hence an increase of $9\dhfar(j')$ in expectation by~\eqref{eq:farthercost}. 
In total, $D_j\leq(4+3\beta)h$ and $X_j$ has expectation at most $4\dhav(j)+10\dhfar(j')$ in this case. 
On the other hand, when $j'$ and $j''$ are matched and they are nearest neighbors of each other, the analysis and upper bounds are the same. 

With all cases taken into consideration, we take the maximum on $D_j$ and $
\E[X_j]$ and obtain
\begin{equation}
D_j=(4+5\beta)h,\,\E[X_j]\leq4\dhav(j)+10\dhfar(j').\label{eq:totalcost}
\end{equation}

We first assume $d(j,j')>\alpha h$ for another parameter $\alpha\in(4,\beta-1]$ to be determined later (recall that $\beta>5$). Then from $\alpha h<d(j,j')\leq 4\dav(j)$ and $\dhfar(j')\leq\dhav(j')\leq\dav(j')\leq \dav(j)\leq\dhav(j)+h$, it is easy to see that 
\[h<\frac{4}{\alpha-4}\cdot\dhav(j)\Rightarrow \dhfar(j')\leq \frac{\alpha}{\alpha-4}\cdot\dhav(j),\]
so in this case, $\E[X_j]$ is bounded by 
\begin{equation}
    \E[X_j]\leq\left(4+\frac{10\alpha}{\alpha-4}\right)\dhav(j)=\frac{14\alpha-16}{\alpha-4}\cdot\dhav(j).\label{eq:firstcase:cost}
\end{equation}

Second, if $d(j,j')\leq\alpha h$, We claim that in the fractional assignment $x$, $j$ is served to an extent by facilities in $F_t\backslash \oball(j',\beta h)$ no less than $j'$ does, otherwise we could shift a positive amount of assignment from $F_t\backslash \oball(j',\beta h)$ into $\oball(j',\beta h)$ for $j'$, resulting in a better assignment of $j'$, which is a contradiction. Therefore, we know for sure that
\[\sum_{i\in F_t\backslash \oball(j',\beta h)}x_{ij'}\leq\sum_{i\in F_t\backslash \oball(j',\beta h)}x_{ij}.\]

We then try to modify the assignment of $j'$ and get a sub-optimal solution, which helps us relate $\dhfar(j')$ to $\dhav(j)$. In the altered assignment $x'$, only assignments of $j'$ are changed, where $x'_{ij'}=x_{ij}$ for every $i\in F_t\backslash \oball(j',\beta h)$. It is easy to see that this is possible, simply by shifting some fractional assignment of $j'$ in $\oball(j',\beta h)$ to the outside. Using triangle inequality, for any $i\in F_t\backslash \oball(j',\beta h)$, $d(i,j)\geq d(i,j')-d(j,j')\geq(\beta-\alpha)h\geq h$, so $d_{-h}(i,j)=d(i,j)$ and
\[\frac{d_{-h}(i,j')}{d_{-h}(i,j)}=\frac{d(i,j')}{d(i,j)}\leq\frac{d(i,j')}{d(i,j')-\alpha h}\leq\frac{\beta h}{\beta h-\alpha h}=\frac{\beta}{\beta-\alpha}.\]

Because $x'$ may not be the optimal assignment for $j'$, one has
\begin{equation}
    \dhfar(j')\leq\sum_{i\in F_t\backslash \oball(j',\beta h)}x'_{ij'}d_{-h}(i,j')\leq\frac{\beta}{\beta-\alpha}\sum_{i\in F_t\backslash \oball(j',\beta h)}x_{ij}d_{-h}(i,j)\leq\frac{\beta}{\beta-\alpha}\cdot\dhav(j).\label{eq:secondcase:cost}
\end{equation}

Combining~\eqref{eq:totalcost}\eqref{eq:firstcase:cost}\eqref{eq:secondcase:cost}, we have
\[D_j\leq (4+5\beta)h,\,\E[X_j]\leq\max\left\{\frac{14\alpha-16}{\alpha-4},\frac{14\beta-4\alpha}{\beta-\alpha}\right\}\dhav(j),\,\textrm{where }\beta>5,4<\alpha\leq\beta-1.\]

Plugging in $\alpha=2+2\sqrt{3},\,\beta=4+2\sqrt{3}$, one has
\[D_j\leq (24+10\sqrt{3})h,\,\E[X_j]\leq(24+10\sqrt{3})\dhav(j),\,\Pr[d(j,A_t)\leq D_j+X_j]=1.\]

Notice that this holds for all clients $j\in C_t$. In this case, for the expectation of $\topl(d(C_t,A_t))$, since we are only paying for at most $\ell$ of the deterministic values (i.e., $D_j$'s), it follows that,
\[\E\left[\topl(d(C_t,A_t))\right]\leq \ell\cdot\max_jD_j+\sum_{j\in C_t}\E[X_j]\leq 41.33\ell\cdot h+41.33\sum_{j\in C_t}\dhav(j).\]
\end{proof}

\begin{proof}[Proof of Lemma \ref{corelem}]
There are $|F_t|\cdot|C_t|$ possible distinct distances at time $t$, hence the number of guesses is obvious, and each guess is associated with a family of reduced cost functions $\mathfrak{D}$, which fully dictates \ref{lp-median} and the running (the non-stochastic part) of our algorithm.

Assume our guessed thresholds are exactly those in the optimal solution in the following. Let $\LP$ be the optimum of our corresponding \ref{lp-median}, $\{A_t\}_{t=1}^T$ be the stochastic output solution, and $\{O_t\}_{t=1}^T$ be the optimal solution with a total cost $\opt=\opt_\mathsf{service}+\opt_\mathsf{move}$, where $\opt_\mathsf{service}$ is the total client service cost and $\opt_\mathsf{move}$ is the total facility movement cost ($\gamma$-scaled). We can easily see $\opt\geq\LP$ by considering the LP solution induced by the optimum. Also define $\sol=\sol_\mathsf{service}+\sol_\mathsf{move}$, where the two parts represent the total client service cost and total facility movement cost ($\gamma$-scaled) in our solution.
	
For each $t\in[T],r\in[N_t]$, let $I_r^\pt$ be the largest index in $w_t$ that has value $w_{tI_r^\pt}=\bar w_{tr}$. Using Lemma \ref{budgetlemma} with $h=T_r^\pt$ and $\ell=I_r^\pt$, we have
\begin{equation}
\E\left[\mathrm{Top}_{I_r^\pt}d(C_t,A_t))\right]\leq41.33I_r^\pt T_r^\pt+41.33\sum_{j\in C_t}\sum_{i\in F_t}d_{-T_r^\pt}(i,j)x_{ij}^\pt,\label{eq:rect:applied}
\end{equation}
and we decompose the ordered cost $w_t^\top d(C_t,A_t)^\downarrow$ as a conic combination of $\topl$ norms (see, e.g.,~\cite{byrka2018constant}), 
\begin{align}
\E[w_t^\top d(C_t,A_t)^\downarrow]&=\sum_{r=1}^{N_t}(\bar w_{tr}-\bar w_{t(r+1)})\cdot\E\left[\mathrm{Top}_{I_r^\pt}d(C_t,A_t)\right]\notag\\
&\leq\sum_{r=1}^{N_t}(\bar w_{tr}-\bar w_{t(r+1)})\left(41.33 I_r^\pt T_r^\pt+41.33\sum_{j\in C_t}\sum_{i\in F_t}d_{-T_r^\pt}(i,j)x_{ij}^\pt\right)\notag\\
&=41.33\sum_{r=1}^{N_t}(\bar w_{tr}-\bar w_{t(r+1)})I_r^\pt T_r^\pt+41.33\sum_{r=1}^{N_t}(\bar w_{tr}-\bar w_{t(r+1)})\sum_{j\in C_t}\sum_{i\in F_t}d_{-T_r^\pt}(i,j)x_{ij}^\pt.\label{eq:conic:expression}
\end{align}
	
Since our guessed thresholds are correct, in the inner product $\opt_t=w_t^\top d(C_t,O_t)^\downarrow$ of the optimum, every weight that is equal to $\bar w_{tr}$ should be multiplied by a distance at least $T_r^\pt$, thus one has
\begin{align}
\opt_t&\geq \sum_{r=1}^{N_t}\bar w_{tr}(I_r^\pt-I_{r-1}^\pt)T_r^\pt\geq \sum_{r=1}^{N_t}\bar w_{tr}I_r^\pt T_r^\pt-\sum_{r=1}^{N_t}\bar w_{tr}I_{r-1}^\pt T_{r-1}^\pt\notag\\
&\geq \sum_{r=1}^{N_t}\left(\bar w_{tr}I_r^\pt T_r^\pt-\bar w_{t(r+1)}I_{r}^\pt T_{r}^\pt\right)=\sum_{r=1}^{N_t}(\bar w_{tr}-\bar w_{t(r+1)})I_{r}^\pt T_{r}^\pt,\label{eq:conic:firstpart}
\end{align}
where $I_0^\pt=0$, $\opt_t$ denotes the service cost of all clients in $C_t$, and thus $\sum_{t\in[T]}\opt_t=\opt_\mathsf{service}$.
	
Further, for the second term in~\eqref{eq:conic:expression}, we have
\begin{align}
\sum_{r=1}^{N_t}(\bar w_{tr}-\bar w_{t(r+1)})\sum_{j\in C_t}\sum_{i\in F_t}d_{-T_r^\pt}(i,j)x_{ij}^\pt
&=\sum_{j\in C_t}\sum_{i\in F_t}x_{ij}^\pt\cdot\sum_{r=1}^{N_t}(\bar w_{tr}-\bar w_{t(r+1)}) d_{-T_r^\pt}(i,j)\notag\\
&=\sum_{j\in C_t}\sum_{i\in F_t}x_{ij}^\pt\cdot d(i,j)\sum_{r:T_r^\pt\leq d(i,j)}^{N_t}(\bar w_{tr}-\bar w_{t(r+1)})\notag\\
&=\sum_{j\in C_t}\sum_{i\in F_t}x_{ij}^\pt\cdot d^\pt(i,j).\label{eq:conic:secondpart}
\end{align}
To see the last equality, notice that the sum is taken over all $r$ s.t. $T_r^\pt\leq d(i,j)$. Since $T_r^\pt$ is decreasing and $T_{N_t+1}^\pt=0$, this sum is equal to $\bar w_{tr_0}$ s.t. $r_0=\min\{r\in[N_t]:T_r^\pt\leq d(i,j)\}$, and $\bar w_{tr_0}d(i,j)=d^\pt(i,j)$ by definition. 
	
Combining~\eqref{eq:conic:expression}\eqref{eq:conic:firstpart}\eqref{eq:conic:secondpart}, our algorithm outputs a stochastic solution $\{A_t\}_{t\in[T]}$ such that the expected service cost of all clients is at most 
\begin{equation}
\E[\sol_\mathsf{service}]=\E\left[\sum_{t=1}^Tw_t^\top d(C_t,A_t)^\downarrow\right]\leq41.33\sum_{t=1}^T\opt_t+41.33\sum_{t=1}^T\sum_{j\in C_t}\sum_{i\in F_t}x_{ij}^\pt\cdot d^\pt(i,j).\label{eq:service:cost}
\end{equation}
	
Meanwhile by Lemma \ref{translemma}, the movement cost of all facilities in $\{A_t\}_{t\in[T]}$ has an expectation of at most \begin{equation}
\E[\sol_\mathsf{move}]\leq\gamma\sum_{t=1}^{T-1}\sum_{i\in F_t}\sum_{i'\in F_{t+1}}d(i,i')z_{ii'}^\pt+6\gamma\cdot\mathbbm{1}[T\geq3]\sum_{t\in [T]}\sum_{j\in C_t}\dav^\pt(j).\label{eq:movement:cost}
\end{equation}
	
By combining~\eqref{eq:service:cost}\eqref{eq:movement:cost}, if $T=2$, our overall cost is 
\begin{align*}
\E[\sol]&\leq41.33\opt_\mathsf{service}+41.33\sum_{t=1}^T\sum_{j\in C_t}\sum_{i\in F_t}x_{ij}^\pt\cdot d^\pt(i,j)+\gamma\sum_{t=1}^{T-1}\sum_{i\in F_t}\sum_{i'\in F_{t+1}}d(i,i')z_{ii'}^\pt\\
&\leq41.33\opt+41.33\sum_{t=1}^T\sum_{j\in C_t}\sum_{i\in F_t}x_{ij}^\pt\cdot d^\pt(i,j)+41.33\gamma\sum_{t=1}^{T-1}\sum_{i\in F_t}\sum_{i'\in F_{t+1}}d(i,i')z_{ii'}^\pt\\
&\leq41.33\opt+41.33\LP\leq82.66\opt.
\end{align*}

For the case where $T\geq3$, recall that we assume the smallest entry in $\{w_t\}_{t\in[T]}$ is at least $\epsilon>0$.
Using Lemma \ref{translemma}, we have the following,
\begin{align*}
\E[\sol]&\leq 41.33\opt_\mathsf{service}+41.33\sum_{t=1}^T\sum_{j\in C_t}\sum_{i\in F_t}x_{ij}^\pt\cdot d^\pt(i,j)+\gamma\sum_{t=1}^{T-1}\sum_{i\in F_t}\sum_{i'\in F_{t+1}}d(i,i')z_{ii'}^\pt\\
&\quad+6\gamma\sum_{t\in [T]}\sum_{j\in C_t}\sum_{i\in F_t}x_{ij}^\pt\cdot d(i,j)\\
&\leq 41.33\opt+41.33\sum_{t=1}^T\sum_{j\in C_t}\sum_{i\in F_t}x_{ij}^\pt\cdot d^\pt(i,j)+\gamma\sum_{t=1}^{T-1}\sum_{i\in F_t}\sum_{i'\in F_{t+1}}d(i,i')z_{ii'}^\pt\\
&\quad+\frac{6\gamma}{\epsilon}\sum_{t\in [T]}\sum_{j\in C_t}\sum_{i\in F_t}x_{ij}^\pt\cdot d^\pt(i,j)\\
&\leq41.33\opt+\left(41.33+\frac{6\gamma}{\epsilon}\right)\LP
\leq\left(82.66+\frac{6\gamma}{\epsilon}\right)\opt,
	\end{align*}	
	where the second inequality is because the smallest entry in $\{w_t\}_{t\in[T]}$ is at least $\epsilon$, thus our reduced cost function $d^\pt$ always satisfies $d^\pt\geq\epsilon d$ for any $t$, since $d^\pt(i,j)=d(i,j)\bar w_{ts}\geq\epsilon d(i,j)$ by definition.
\end{proof}

\begin{proof}[Proof of Theorem \ref{thm:dom}]
To start with, we fix a small constant $\delta>0$ and consider a slightly modified instance, where instead of using the weight vector $w_t$, we define $\tw_t$ where $\tw_{ts}=\max\{w_{ts},\delta w_{t1}/|C_t|\}$. 
It is easy to see that for any \emph{non-negative non-increasing} vector $\vcv\in\mathbb{R}_{\geq0}^{|C_t|}$, one has $\tw_t^\top\vcv\geq w_t^\top\vcv$, and $\tw_t^\top\vcv-w_t^\top\vcv\leq \delta w_{t1}\vcv_1\leq\delta w_t^\top\vcv$, which can be rewritten as $\tw_t^\top\vcv\in[w_t^\top\vcv,(1+\delta)w_t^\top\vcv]$. 
Therefore, by solving $T$-\dom on these new weight vectors, we lose a factor of at most $1+\delta$ in the approximation ratio. Suppose we do this in the sequel.

Since $T=O(1)$, we also assume that the maximum connection distance $T_{\max}^\pt\geq0$ is known to us for any $t\in[T]$ in a fixed optimal solution.
This is done via simple exhaustive search, and we only lose a polynomial factor in the running time. 
Fix $t$ and let $M_t=\ceil{\log_{1+\delta}(|C_t|/\delta)}$.
For $r=1,2,\dots,M_t$, define the real intervals 
\[J_r=\left((1+\delta)^{-(M_t-r+1)}T_{\max}^\pt,(1+\delta)^{-(M_t-r)}T_{\max}^\pt\right],\] 
and let $J_0=[0,(1+\delta)^{-M_t}T_{\max}^\pt]$. 
We have that $\calJ=\{J_0,\dots,J_{M_t}\}$ is a partition of the interval $[0,T_{\max}^\pt]$, i.e., the interval that all connection distances at $t$ fall in, with respect to the optimum.

Denote $\vco$ the \emph{sorted} service cost vector at time $t$ in the optimum, and $I_{J}=\{s\in[|C_t|]:\vco_s\in J\}$ the indices of entries of $\vco$ in $J$. 
Obviously, $I_{J_{M_t}}$ is non-empty, because our guess $T_{\max}^\pt$ is correct and $T_{\max}^\pt\in J_{M_t}$. 
For $r=M_t,M_t-1,\dots,1,0$, iteratively define the average weight in the following way (we omit the superscript $(t)$ here for simplicity),
\begin{equation}
    \bw_{r}=\left\{
    \begin{array}{cc}
    \frac{1}{|I_{J_r}|}\sum_{r\in I_{J_r}}\tw_{tr} & I_{J_r}\neq\emptyset\\
    \bw_{r+1} & I_{J_r}=\emptyset,
    \end{array}
    \right.
\end{equation}
and this is well-defined since $I_{J_{M_t}}$ is always non-empty. 
Additionally, note that $\bw=\{\bw_0,\dots,\bw_{M_t}\}$ is non-decreasing and they are all in the interval $[\delta\tw_{t1}/|C_t|,\tw_{t1}]$. 
Although we have no knowledge of the optimum or the exact values of these average weights, it is possible for us to approximately guess their values. 
To achieve this, recall that $M_t=\ceil{\log_{1+\delta}(|C_t|/\delta)}$ and all entries of $\bw$ are in the interval $[\delta\tw_{t1}/|C_t|,\tw_{t1}]$.
Therefore, if we guess each $\bw_r$ to its \emph{closest and no smaller} power of $1+\delta$ (hence $\hw_r\in[\bw_r,(1+\delta)\bw_r]$), there are at most $2+\log_{1+\delta}(|C_t|/\delta)$ possible candidates, and they must also be non-decreasing as $\bw$ does. Using a basic counting method, the number of such non-decreasing sequences is at most $\exp(O(\log_{1+\delta}(|C_t|/\delta)))=(|C_t|/\delta)^{O(1/\delta)}$, thus bounded by a polynomial.

We define the reduced cost functions now. For each $t\in[T]$, we guess the maximum connection distance $T_{\max}^\pt$ and the approximate values of average weights $\hw^\pt=\{\hw_0^\pt,\dots,\hw_{M_t}^\pt\}$. 
Suppose the guesses are correct in what follows. 
The reduced cost function $d^\pt$ is defined as $d^\pt(i,j)=\hw_r^\pt\cdot d(i,j)$ where $d(i,j)\in J_r$, $r=0,\dots,M_t$. 
We notice that $d^\pt$ is undefined for distances larger than $T_{\max}^\pt$, and this is easily handled by explicitly adding to \ref{lp-median} the constraints $x_{ij}^\pt=0$ for all $i\in F_t,j\in C_t$ s.t. $d(i,j)>T_{\max}^\pt$. 
Fix $t$ in the following, and for each $s\in[|C_t|]$, define $T_s=\sup(J)$ where $\vco_s\in J$ and $J\in\calJ$. 
We obviously have $(1+\delta)\vco_s\geq T_s$ if $s\geq1$, and otherwise we have $\vco_s\in J_0$ and thus $\vco_s\leq T_s\leq\delta T_{\max}^\pt/|C_t|$. 
Using Lemma \ref{budgetlemma} with $h=T_s$ and $\ell=s$, $s=1,2,\dots,|C_t|$,
\begin{align}
    \E[\tw_t^\top d(C_t,A_t)^\downarrow]&=\sum_{s=1}^{|C_t|}(\tw_{ts}-\tw_{t(s+1)})\E[\mathrm{Top}_{s}(d(C_t,A_t))]\notag\\
    &\leq\sum_{s=1}^{|C_t|}(\tw_{ts}-\tw_{t(s+1)})\left(41.33s\cdot T_s+41.33\sum_{j\in C_t}\sum_{i\in F_t}d_{-T_s}(i,j)x_{ij}^\pt\right)\notag\\
    &\leq41.33(1+\delta)\sum_{s:\vco_s\in J_{\geq1}}(\tw_{ts}-\tw_{t(s+1)})s\cdot\vco_s
    +41.33\sum_{s:\vco_s\in J_0}(\tw_{ts}-\tw_{t(s+1)})s\cdot\frac{\delta T_{\max}^\pt}{|C_t|}\notag\\
    &\quad+41.33\sum_{s=1}^{|C_t|}(\tw_{ts}-\tw_{t(s+1)})\sum_{j\in C_t}\sum_{i\in F_t}d_{-T_s}(i,j)x_{ij}^\pt\notag\\
    &\leq41.33(1+\delta)\sum_{s=1}^{|C_t|}(\tw_{ts}-\tw_{t(s+1)})\mathrm{Top}_s(\vco)+41.33\delta\tw_{t1}T_{\max}^\pt\notag\\
    &\quad+41.33\sum_{s=1}^{|C_t|}(\tw_{ts}-\tw_{t(s+1)})\sum_{j\in C_t}\sum_{i\in F_t}d_{-T_s}(i,j)x_{ij}^\pt\notag\\
    &\leq41.33(1+2\delta)\tw_t^\top\vco+41.33\sum_{s=1}^{|C_t|}(\tw_{ts}-\tw_{t(s+1)})\sum_{j\in C_t}\sum_{i\in F_t}d_{-T_s}(i,j)x_{ij}^\pt.\label{eq:true:cost:strong}
\end{align}

To bound~\eqref{eq:true:cost:strong}, we notice that its first part is exactly $41.33(1+2\delta)\opt_t$, where $\opt_t$ is the total ordered service cost paid in the optimum at time $t$. For the second part, we obtain
\begin{align}
\sum_{s=1}^{|C_t|}(\tw_{ts}-\tw_{t(s+1)})\sum_{j\in C_t}\sum_{i\in F_t}d_{-T_s}(i,j)x_{ij}^\pt
&=\sum_{j\in C_t}\sum_{i\in F_t}\sum_{s=1}^{|C_t|}(\tw_{ts}-\tw_{t(s+1)})x_{ij}^\pt d_{-T_s}(i,j)\notag\\
&=\sum_{j\in C_t}\sum_{i\in F_t}x_{ij}^\pt d(i,j)\sum_{s\in[|C_t|],T_s\leq d(i,j)}(\tw_{ts}-\tw_{t(s+1)}).\label{eq:second:cost:strong}
\end{align}

We group all $d(i,j)$ s.t. $d(i,j)\in J_0$ together, and notice that $d(i,j)\leq\delta T_{\max}^\pt/|C_t|$. Since we have $\sum_{i\in F_t}x_{ij}^\pt=1$ for all $j\in C_t$, the total contribution of all such pairs $(i,j)$ in~\eqref{eq:second:cost:strong} is at most $\delta\tw_{t1}T_{\max}^\pt\leq\delta\opt_t$. 
Now suppose $d(i,j)\in J_{r_0},\,r_0\geq1$, and for a technical reason, we assume that none of the input distances coincides with any boundary of the intervals. 
This is achieved by slightly increasing the boundaries of all intervals (also see~\cite{byrka2018constant} about this trick). 
If there is no such $T_s\leq d(i,j)$, the corresponding contribution in~\eqref{eq:second:cost:strong} is zero, while $d^\pt(i,j)=\hw_{r_0}\cdot d(i,j)$ is non-negative; otherwise, the sum of weights is the maximum $\tw_{ts}$ s.t. $\vco_s\in J_{<r_0}$, which is at most the average weight on $J_{r_0}$ and hence at most $\bw_{r_0}\leq\hw_{r_0}$. 
This implies that~\eqref{eq:second:cost:strong} is at most
\begin{equation}
\sum_{s=1}^{|C_t|}(\tw_{ts}-\tw_{t(s+1)})\sum_{j\in C_t}\sum_{i\in F_t}d_{-T_s}(i,j)x_{ij}^\pt\leq\sum_{j\in C_t}\sum_{i\in F_t}x_{ij}^\pt d^\pt(i,j)+\delta\opt_t,\label{eq:third:cost:strong}
\end{equation}
and the first part is exactly the service cost of time step $t$ in the objective of \ref{lp-median}.

Finally, via considering the solution $(x^\star,y^\star,z^\star)$ induced by the optimum of the original problem, the optimum of the relaxation \ref{lp-median} under the aforementioned reduced cost functions is at most,
\begin{align}
    \LP&\leq\sum_{t=1}^T\sum_{j\in C_t}\sum_{i\in F_t}d^\pt(i,j)x_{ij}^{\star\pt} +\gamma\sum_{t=1}^{T-1}\sum_{i\in F_t}\sum_{i'\in F_{t+1}}d(i,i')z_{ii'}^{\star\pt}\notag\\
    &\leq\sum_{t=1}^T\sum_{s=1}^{|C_t|}\hw_{r_s}^\pt\cdot\vco_s^\pt+\opt_{\mathsf{move}}\notag\\
    &\leq(1+\delta)\sum_{t=1}^T\sum_{r=0}^{|M_t|}\bw_{r}^\pt\sum_{s\in I_{J_r}^\pt}\vco_s^\pt+\opt_{\mathsf{move}}\notag\\
    &\leq(1+\delta)\sum_{t=1}^T\sum_{r:|I_{J_r}^\pt|\neq\emptyset}
    \left(\sum_{s\in I_{J_r}^\pt}\tw_{ts}\right)\left(\frac{1}{|I_{J_r}^\pt|}
    \sum_{s\in I_{J_r}^\pt}\vco_s^\pt\right)+\opt_{\mathsf{move}}\notag\\
    &\leq(1+\delta)^2\sum_{t=1}^T\sum_{r:|I_{J_r}^\pt|\neq\emptyset}
    \sum_{s\in I_{J_r}^\pt}\tw_{ts}\vco_s^\pt+(1+\delta)\sum_{t=1}^T\delta\tw_{t1}T_{\max}^\pt+\opt_{\mathsf{move}}\notag\\
    &\leq (1+3\delta+2\delta^2)\sum_{t=1}^T\tw_t^\top\vco^\pt+\opt_{\mathsf{move}}\leq(1+3\delta+2\delta^2)\opt,\label{eq:lp:cost:strong}
\end{align}
where $\hw_{r_s}^\pt$ denotes the guessed average weight where $\vco_s^\pt\in J_{r_s}$. We note that we again consider $J_0$ and $J_{>0}$ separately. For the former, the total contribution is easily bounded using $\sup(J_0)\leq\delta T_{\max}^\pt/|C_t|$; for the latter, since the distances falling into the same interval differ by a multiplicative factor of at most $\delta$, we have the above inequality.

By combining~\eqref{eq:true:cost:strong}\eqref{eq:third:cost:strong}\eqref{eq:lp:cost:strong} and taking the sum over $t\in[T]$, the rest follows from the same analysis in Lemma \ref{corelem}, thus omitted here.
\end{proof}


\section{Approximating \dks}\label{section-outlier}

In this section, we present our various results on \dks. We start by showing two hardness results for $T$-\dks and $T$-\rdks when $T\geq3$, then present a simple 3-approximation for 2-\dks and another more involved multi-criteria approximation algorithm for 2-\rdks in the more nuanced outlier setting. We finish this section with a hardness result on pure approximations for 2-\rdks.

\subsection{The hardness of approximating \emph{3-\dks} and \emph{3-\rdks}}\label{neg:3dks}

As a warm-up, we use a simple argument to show that when the number of time steps is at least 3, $T$-\dks admits no non-trivial approximations unless $\mathrm{P=NP}$. The proof uses the reduction from the \emph{perfect 3D-matching} problem, which is known to be NP-complete~\cite{karp2010reducibility}.

\begin{theorem}\label{thm:dksnphard}
If $T\geq3$, there is no polynomial-time algorithm for \emph{$T$-\dks} with non-trivial approximation factors, unless $\mathrm{P=NP}$.
\end{theorem}
\begin{proof}
Notice that we only need to prove the hardness of 3-\dks, since this is a special case of $T>3$ by setting $C_t=\emptyset,\,t\geq4$. We reduce an arbitrary instance of perfect 3D-matching to 3-\dks. Recall that in an instance of perfect 3D-matching, we are given three finite ground sets $A,B,C$ with $|A|=|B|=|C|$, and a triplet set $\calT\subseteq A\times B\times C$. Suppose $|A|=n$ and $|\calT|=m$, and one needs to decide whether there exists a subset $\calS\subseteq\calT$, such that $|\calS|=n$, and each element in $A\cup B\cup C$ appears exactly once in some triplet in $\calS$. W.l.o.g., we assume that each element in $A\cup B\cup C$ appears in at least one triplet, otherwise the answer is trivially negative.

Assume there is an $\alpha$-approximation for 3-\dks with factor $\alpha>1$. We construct the following graph $G=(V_A\cup V_B\cup V_C,E)$, where the vertex set and edge set are initially empty.
\begin{enumerate}[Step 1.]
	\item For each triplet $g=(a,b,c)\in\calT$ where $a\in A,\,b\in B,\,c\in C$, add three new vertices $V_A\leftarrow V_A\cup\{a_g\},V_B\leftarrow V_B\cup\{b_g\},V_C\leftarrow V_C\cup\{c_g\}$. Connect $(a_g,b_g)$ with an edge of length $\alpha$ and add it to $E$. Connect $(b_g,c_g)$ with an edge of length $\alpha$ and add it to $E$.
	\item For any two vertices $a_g\in V_A,\,a_{g'}\in V_A$ corresponding to the same element $a\in A$, connect them with an edge of length 1 and add the edge to $E$, thus forming a clique $K_a$ with unit-length edges. Repeat the same procedure for $V_B,\,V_C$.
\end{enumerate}
	
An illustration is shown in Fig. \ref{figure:matching}. We approximately solve 3-\dks on $G$ with its graph metric $d_G$ (set $d_G(s,t)=\infty$ if they are in different connected components), with $k=n$ and the movement constraint $B=\alpha$, where the client sets are $C_1=V_A,\,C_2=V_B,\,C_3=V_C$ and facility sets are $F_1=V_A,\,F_2=V_B,\,F_3=V_C$. W.l.o.g., this 3-\dks instance has optimum at least 1, otherwise one has $m=n$ and the original instance is again trivial.
    
Now, if the original perfect 3D-matching instance is feasible, by letting each of the $n$ mobile facilities move along the $n$ trajectories induced by the feasible solution, we have a solution to the 3-\dks instance on $G$ with objective exactly 1, since each clique $K_u$ corresponding to some $u\in A\cup B\cup C$ is formed using unit-length edges, and there is exactly one open facility in every such clique. Conversely, if the 3-\dks instance has optimum 1, it is easy to see that there must be exactly one vertex chosen as an open facility in each clique $K_u,\,u\in A\cup B\cup C$, and since the movement constraint is $B=\alpha$, each open facility must follow the edges $(a_g,b_g)$ and $(b_g,c_g)$ for some $g\in\calT$ when it moves. Thus, this optimum directly induces a feasible solution to the original perfect 3D-matching instance.
    
By noticing that for any $s\neq t\in V_A$, either $d_G(s,t)=1$ or $d_G(s,t)\geq2\alpha$, one has that the optimum of the 3-\dks instance is either equal to 1 or no less than $2\alpha$. Hence any $\alpha$-approximation for 3-\dks implies an efficient algorithm that determines the feasibility of the perfect 3D-matching instance, which is prohibited by its NP-completeness.
\end{proof}

\begin{figure}[hbt!]
    \centering
    \includegraphics{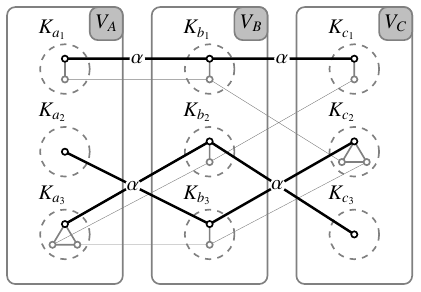}
    \caption{An illustration of the reduction from perfect 3D-matching to 3-\dks, with the instance $\calT=\{(a_1,b_1,c_1),(a_2,b_3,c_2),(a_3,b_2,c_3),(a_1,b_1,c_2),(a_3,b_2,c_1),(a_3,b_3,c_2)\}$, $A=\{a_1,a_2,a_3\}$, $B=\{b_1,b_2,b_3\}$ and $C=\{c_1,c_2,c_3\}$. 
    Every intra-clique edge has length 1, and we highlight a feasible solution with objective value 1.}
    \label{figure:matching}
\end{figure}

From the result above, the hardness of pure approximations for $T$-\rdks with $T\geq3$ directly follows, since $T$-\dks is only its special case. We strengthen this result by showing that, even by allowing multi-criteria approximations, i.e., violating the outlier constraints by some small $\epsilon$-fraction, one is still unable to obtain any non-trivial approximation factors for $T$-\rdks, $T\geq3$. The result is formally given in the following theorem, where we use the reduction from \emph{maximum 3D-matching}, which is known to be APX-complete~\cite{chlebik2006complexity,kann1991maximum}.

\begin{theorem}\label{rdksapxhard}
    If $T\geq3$, there exists a constant $\epsilon_0\in(0,1)$, such that \emph{$T$-\rdks} admits no multi-criteria $(\alpha,1-\epsilon_0,\dots,1-\epsilon_0)$-approximations for any non-trivial factor $\alpha>1$, unless $\mathrm{P=NP}$.
\end{theorem}
\begin{proof}
    Again, we only need to prove the case of 3-\rdks, since this is a special case of larger $T$'s. For a maximum 3D-matching instance on ground sets $A,\,B,\,C$ (not necessarily with identical cardinalities) and a triplet set $\calT\subseteq A\times B\times C$, we are asked to find $\calS\subseteq\calT$ such that the triplets in $\calS$ are pairwise disjoint and $|\calS|$ is maximized.
    
    By way of contradiction, assume we have for \emph{some} $\alpha>1$ and \emph{any} constant $\epsilon>0$, an efficient multi-criteria $(\alpha,1-\epsilon,1-\epsilon,1-\epsilon)$-approximation for 3-\rdks. Given any maximum 3D-matching instance, we first create the same graph $G=(V_A\cup V_B\cup V_C,E)$ as in the proof for Theorem \ref{thm:dksnphard}. Let $m=\max_{u\in A\cup B\cup C}|K_u|$ be the maximum size of the cliques in $G$, i.e., the maximum number of occurrences of elements in the triplets, and add additional dummy vertices to each clique so that each clique has size exactly $m$. Evidently, any intra-clique distance is 1 and any inter-clique distance is at least $2\alpha$ (within the same time step).
    
    Now, we run the approximation algorithm for 3-\rdks on $G$ with $C_1=F_1=V_A$, $C_2=F_2=V_B$ and $C_3=F_3=V_C$ and movement constraint $B=\alpha$ as the input. The algorithm is run for multiple times, for the parameter $k=1,2,\dots,\min\{|A|,|B|,|C|\}$ and outlier constraints $l_1=l_2=l_3=mk$ each time.
    
    Suppose the optimum for the maximum 3D-matching instance is $k_0\in\mathbb{Z}_+$, and it is easy to see that this optimum induces a solution to the 3-\rdks instance on $G$ that utilizes $k_0$ open facilities and covers exactly $mk_0$ clients within distance 1 at each time step, using intra-clique unit-length edges. Because our algorithm is an $(\alpha,1-\epsilon,1-\epsilon,1-\epsilon)$-approximation, when we run it with parameter $k=k_0$, it must be able to output a solution that uses $k_0$ open facilities and covers at least $(1-\epsilon)mk_0$ clients within distance $\alpha$ at each time step, and this solution naturally induces a subset $\calS\subseteq\calT$ with $|\calS|=k_0$, according to the definition of $G$. Recall that inter-clique distances are at least $2\alpha$, hence this solution must use only intra-clique edges. We say that $u\in A\cup B\cup C$ is ``matched'' in this solution, if and only if some vertex in the clique $K_u$ is selected as an open facility location. Thus the numbers of elements in $A,B$ and $C$ that are matched in this 3-\rdks solution are all at least $(1-\epsilon)k_0m/m=(1-\epsilon)k_0$, respectively. 
    
    This means that $\calS$ with $|\calS|=k_0$, which may have intersecting triplets, matches at least $(1-\epsilon)k_0$ elements in $A,B$ and $C$. To further distill the induced subset $\calS\subseteq\calT$ and remove intersecting triplets, we iteratively remove some arbitrary $g\in\calS$ and assign $\calS\leftarrow\calS\setminus\{g\}$, whenever $g$ intersects some other triplet in $\calS$. It is easy to see that, at most $3\epsilon k_0$ triplets are removed in this process, and we obtain $|\calS|\geq(1-3\epsilon)k_0$ as a feasible 3D-matching solution. This happens for any constant $\epsilon>0$, hence a de facto PTAS for maximum 3D-matching. But unless $\mathrm{P=NP}$, this is impossible since maximum 3D-matching is APX-complete. Therefore, our initial assumption is incorrect, which yields the theorem as its contrapositive assertion.
\end{proof}

\begin{remark}
One can easily see that the two hardness results above are also valid if we only allow the \dks and \rdks solutions to consist of subsets of open facilities at each time step, instead of multi-sets.
\end{remark}

\subsection{A 3-approximation for \emph{2-\dks}}
\label{pos:2dks}

In contrast to the hardness of approximating $T$-\dks for $T\geq 3$, we consider 2-\dks on general metrics and present a simple flow-based 3-approximation. Suppose that we have successfully guessed the optimum $R^\star$ (using binary search). We construct the following network flow instance $\calG(\calV,\calE)$.
$\calV$ consists of $4$ layers of vertices
(two layers $\calL^{11},\calL^{12}$ for $t=1$, 
two layers $\calL^{21},\calL^{22}$ for $t=2$), 
a source $\source$ and sink $\sink$.
We define $\calG$ as follows:
\begin{enumerate}[Step 1.]
	\item For each $i\in F_1$, add a vertex in $\calL^{12}$. For $i'\in F_2$, add a vertex in $\calL^{21}$.
	\item Repeatedly pick an arbitrary client $j\in C_1$ and remove from $C_1$ every client within distance $2R^\star$ from $j$. Denote these clients a new cluster corresponding to $j$ and call $j$ a cluster center.
	Since $R^\star$ is optimal, it is easy to see we obtain at most $k$ such clusters (otherwise there exist $k+1$ clients with pair-wise distance $>2R^\star$, and every open facility can obviously cover at most one of them, which is a contradiction).
	And if there are less than $k$ clusters, we create some extra dummy clusters to obtain
	exactly $k$ clusters, while dummy clusters do not represent any clients.
	For each cluster, add a vertex to $\calL^{11}$. Repeat this for $C_2$ and form $\calL^{22}$.
	\item The four layers are arranged in order as $\calL^{11},\calL^{12},\calL^{21},\calL^{22}$. With a slight abuse of notation, for a non-dummy cluster center $u\in\calL^{11}$ and facility location $v\in\calL^{12}$, connect them using a link with unit capacity if $d(u,v)\leq R^\star$; for facility location $w\in\calL^{21}$ and a non-dummy cluster center $z\in\calL^{22}$, connect them using a link with unit capacity if $d(w,z)\leq R^\star$. For $v\in\calL^{12},w\in\calL^{21}$ (both are facility locations), connect them using a link with \emph{unbounded} capacity if $d(v,w)\leq B$. 
	\item Connect every dummy cluster in $\calL^{11}$ with every facility location vertex in $\calL^{12}$. Connect every dummy cluster in $\calL^{22}$ with every facility location vertex in $\calL^{21}$. Every such link has unit capacity.
	\item Finally, the source $\source$ is connected to every vertex in $\calL^{11}$
and the sink $\sink$ is connected to every vertex in $\calL^{22}$, with every edge having unit capacity.
\end{enumerate} 

\begin{lemma}\label{lem:2supflow}
$\calG(\calV,\calE)$ admits a flow of value $k$. Moreover, we can obtain a feasible solution of cost at most $3R^\star$ from an integral flow of value $k$ in $\calG(\calV,\calE)$.
\end{lemma}
\begin{proof} Consider an optimal solution $A_1\subseteq F_1,A_2\subseteq F_2$ with objective $R^\star$, and there exists a perfect matching between the two. For any $i\in F_1,i'\in F_2$, if the pair $(i,i')$ appears $m$ times in the perfect matching, define a flow value $f(i,i')=m$ over link $(i,i')$. Notice that this is well-defined, since one must have $d(i,i')\leq B$ and the corresponding vertices are linked in the network.

Consider the first time step. For any facility location $i$ and two non-dummy cluster centers $j,\,j'$, either $d(i,j)$ or $d(i,j')$ is larger than $R^\star$, otherwise $d(j,j')\leq 2R^\star$ using triangle inequality, contradicting with our construction. This means that the subsets of facility locations linked with the non-dummy clusters are disjoint. Next, since $A_1$ covers all $j\in C_1$ with radius $R^\star$, for every cluster center $j\in \calL^{11}$, we can always find a distinct open facility $i\in A_1$ such that $d(i,j)\leq R^\star$, and add a unit flow as $f(j,i)=1$. The same process is repeated for $\calL^{22}$ and $A_2$.

The total flow between $\calL^{12}$ and $\calL^{21}$ is now obviously $k$, since the perfect matching between $A_1$ and $A_2$ has size $k$. After the construction of unit flows for non-dummy clusters, we need to satisfy the flow conservation constraint at all facility vertices. Arbitrarily direct the remaining flows from facility vertices to \emph{dummy} clusters, one unit each time, and this always satisfies the flow conservation constraint at facility vertices since the total numbers of clusters are both $k$ in $\calL^{11}$ and $\calL^{22}$. Finally, for every cluster with unit flow, define the flow value between it and the source/sink as 1. This completes an integral flow of value $k$ on $\calG$.

For the second assertion in the lemma, suppose we have an integral flow $\intflow$ of value $k$ on $\calG$. For any facility location $i\in F_1$, denote $g(i)$ the total flow through $i$. We place $g(i)$ facilities at location $i$, and repeat the same procedures for $i'\in F_2$. If $\intflow(i,i')=m$ for $i\in F_1,i'\in F_2$, move $m$ facilities from $i$ to $i'$ in the transition between the two time steps.

For every $j'\in C_1$, if $j$ is the cluster center it belongs to, using triangle inequality and the fact that $j$ has unit flow on its vertex, the nearest open facility for $j$ is at a distance at most $R^\star$, and there exists an open facility at most $d(j',i)\leq d(j',j)+d(j,i)\leq 3R^\star$ away from $j'$.
\end{proof}

\begin{theorem}\label{thm:dks}
	There exists a 3-approximation for \emph{2-\dks}.
\end{theorem}

\begin{proof}
	Consider the aforementioned network flow instance. It only has integer constraints and the coefficient matrix is totally unimodular. Moreover, there exists a flow of value $k$ due to Lemma \ref{lem:2supflow}, hence we can efficiently compute an integral flow $\intflow$ of value $k$, thus obtaining a 3-approximate solution.
\end{proof}

\subsection{The hardness of pure approximations for \emph{2-\rdks}}

In this section, we show that 2-\rdks admits no pure approximations for any non-trivial factor $\alpha>1$, i.e., multi-criteria $(\alpha,1,1)$-approximation algorithms, complementing our multi-criteria $(3,1-\epsilon,1-\epsilon)$-approximation in the next section for every $\epsilon>0$.

\begin{theorem}\label{2rdkshard}
    There is no polynomial-time multi-criteria $(\alpha,1,1)$-approximation algorithm for \emph{2-\rdks} for any non-trivial factor $\alpha$, unless $\mathrm{P=NP}$.
\end{theorem}
\begin{proof}
    We use the reduction from maximum satisfiability (MAX-SAT), which is known to be NP-hard. By way of contradiction, assume there exists a pure $\alpha$-approximation for 2-\rdks. For an arbitrary SAT instance with boolean variables $x_1,x_2,\dots,x_n$ and $m$ clauses $s_1\wedge s_2\wedge\cdots\wedge s_m$ in conjunctive normal form (CNF), we create a quadripartite graph $G$ and subsequently define the 2-\rdks instance on the graph using its graph metric $d_G$.
    
    $G$ is initially empty. First create the client sets by adding $n+m$ vertices to $G$, where $C_1=\{x_1,\dots,x_n\}$ and $C_2=\{s_1,\dots,s_m\}$, with a slight abuse of notation. For the facility sets, add $4n$ new vertices to $G$, where $F_1=\{y_{1,0},y_{1,1},\dots,y_{n,0},y_{n,1}\}$ and $F_2=\{z_{1,0},z_{1,1},\dots,z_{n,0},z_{n,1}\}$. Here, $y_{i,0}$ and $z_{i,0}$ represent the choice of setting the variable $x_i=0$, and $y_{i,1},\,z_{i,1}$ are similarly defined. 
    
    With the $5n+m$ vertices added to $G$, we proceed to define the edges. For each variable $x_i$, link $(x_i,y_{i,0})$ and $(x_i,y_{i,1})$ using an edge with length 1. For each variable $x_i$, also link $(y_{i,0},z_{i,0})$ and $(y_{i,1},z_{i,1})$ using an edge with length $\alpha$. Finally, for each clause $s_j$ and each literal in $s_j$, link $s_j$ to the corresponding ``literal vertex'' in $F_2$ that satisfies $s_j$, using an edge with length 1. For example, for clause $s_j=x_1\vee x_2\vee\neg x_3$, we link $(s_j,z_{1,1})$, $(s_j,z_{2,1})$ and $(s_j,z_{3,0})$. Also see an illustration in Fig. \ref{figure:boolean}. Evidently, any facility-client connection distance is either 1 or $\geq2\alpha+1$ on this graph metric, within the same time step.
    
    We run our pure $\alpha$-approximation on $G$, with the client sets and facility sets defined as above, and $k=n$, $B=\alpha$. For the outlier constraints, we fix $l_1=n$ and try $l_2=1,2,\dots,m$. Suppose the optimum for the MAX-SAT instance is $m_0$, and it is easy to see that, this optimum induces a solution to the 2-\rdks instance with objective 1 when we have $l_2=m_0$. Since any facility-client connection is either 1 or $\geq2\alpha+1$, our pure $\alpha$-approximation in fact solves MAX-SAT exactly by trying all possible values of $l_2$, in particular when $l_2=m_0$. But this is impossible unless $\mathrm{P=NP}$, hence our initial assumption is incorrect, yielding the theorem.
\end{proof}

\begin{figure}[hbt!]
    \centering
    \includegraphics{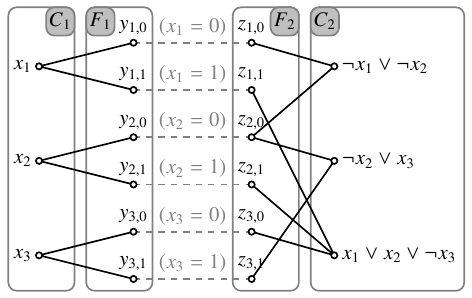}
    \caption{An illustration of the reduction from MAX-SAT to 2-\rdks, where the boolean formula in CNF is $(\neg x_1\vee\neg x_2)\wedge(\neg x_2\vee x_3)\wedge(x_1\vee x_2\vee\neg x_3)$. We use solid lines to represent edges with length 1, and dashed lines for edges with length $\alpha$.}
    \label{figure:boolean}
\end{figure}

As a simple corollary, one can also show the hardness of \emph{guaranteed} multi-criteria $(\alpha,1,1-\epsilon)$ or $(\alpha,1-\epsilon,1)$-approximations for 2-\rdks with some non-trivial $\alpha$ and any constant $\epsilon>0$, using the identical reduction from the above and the fact that MAX-SAT is APX-hard~\cite{arora1998proof}. Here, a guaranteed multi-criteria approximation means that the factors are strictly aligned with the time steps, regardless of the given instance.

\begin{corollary}
    There exists a constant $\epsilon_0\in(0,1)$, such that \emph{2-\rdks} admits no guaranteed multi-criteria $(\alpha,1,1-\epsilon_0)$ or $(\alpha,1-\epsilon_0,1)$-approximations for any non-trivial factor $\alpha>1$, unless $\mathrm{P=NP}$.
\end{corollary}

\subsection{A multi-criteria approximation for \emph{2-\rdks}}

In this section, we present our more involved multi-criteria $(3,1-\epsilon,1-\epsilon)$-approximation for 2-\rdks.
At a high level, the algorithm first guesses the optimal objective via binary search. 
It then guesses a small portion of the unknown optimal solution that ``covers'' as many clients as possible (Section \ref{section:dksout:guess}), and attempts to formulate the remaining problem as finding a bipartite matching that attains good coverage of clients on both sides (Section \ref{section:dksout:matching}).

Before we start, we formulate the problem in an alternative and more convenient way. 
In the original instance, the client sets $C_1,C_2$ and candidate facility locations $F_1,F_2$ all live in the same finite metric space $(X,d)$, while a facility can move to any other location as long as the distance is no more than some given threshold $B$. 
In this section, we create two copies $(X^\pone,d)$ and $(X^\ptwo,d)$ of the original space $(X,d)$, and consider $C_1,F_1\subseteq X^\pone$, $C_2,F_2\subseteq X^\ptwo $ instead. 
Now, a facility $i\in F_1$ can move to $i'\in F_2$ if and only if their corresponding locations in $X$ satisfy $d(i,i')\leq B$. 
Obviously, this formulation is equivalent to the original one. 
In the sequel, we also assume that we have correctly guessed the optimum $R^\star$ via binary search, since it only has a polynomial number of possibilities.

\subsubsection{Guessing and modifying the instance}
\label{section:dksout:guess}

We start with some pre-processing, including some must-have facility locations and clients and excluding some others. The intuition is to guess a constant-size ``heavy'' subset of the optimal solution, and consider the problem on the reduced instance. Denote $B_1(i,R)=\{j\in C_1:d(i,j)\leq R\},\,B_2(i,R)=\{j\in C_2:d(i,j)\leq R\}$. Fix some constant $\gamma>0$ (its value will be determined later), we do the following to modify the original instance.
\begin{enumerate}[Step 1.]
	\item Enumerate all possible choices of $\gamma/\epsilon$ distinct facilities in $F_1$ and $\gamma/\epsilon$ distinct facilities in $F_2$. Denote the two chosen sets of facilities $T_1,\,T_2$. Additionally, enumerate all possible movements associated with $T_1$ and $T_2$, i.e., $g:T_1\rightarrow F_2$ and $h:T_2\rightarrow F_1$ such that $\forall i\in T_1,d(i,g(i))\leq B,\,\forall i'\in T_2,d(i',h(i'))\leq B$.
	\item Recursively sort $T_1$, each time by choosing the unchosen $i\in T_1$ such that $B_1(i,3R^\star)$ covers the largest number of uncovered clients in $C_1$, and letting this number be $u_i^\pone$. Denote $u_0^\pone$ the number of remaining clients that are covered by $\bigcup_{i\in h(T_2)}B_1(i,3R^\star)$. Same for $T_2$ and $C_2$, and we have $u_0^\ptwo $ and $u_i^\ptwo $ for $i\in T_2$. Denote $C_1'=C_1\setminus\left(\bigcup_{i\in T_1\cup h(T_2)}B_1(i,3R^\star)\right)$ and $C_2'=C_2\setminus\left(\bigcup_{i\in T_2\cup g(T_1)}B_2(i,3R^\star)\right)$, i.e., the clients that are not covered by any of the closed balls above.
	\item Remove any $i'\in F_1$ s.t. $|B_1(i',3R^\star)\cap C_1'|>\min_{i\in T_1}\{u_i^\pone\}$ and any $i'\in F_2$ s.t. $|B_2(i',3R^\star)\cap C_2'|>\min_{i\in T_2}\{u_i^\ptwo \}$. Denote the remaining two facility sets $F_1'$ and $F_2'$. 
	Notice that each removed $i'\in F_1$ is \emph{not} co-located with $T_1\cup h(T_2)$, and each removed $i'\in F_2$ is \emph{not} co-located with $T_2\cup g(T_1)$;
	\item Denote the reduced problem $\calP'$, where the client sets are $C_1',\,C_2'$, the facility sets are $F_1',\,F_2'$ and the outlier constraints are $l_1'=\max\left\{l_1-u_0^\pone-\sum_{i\in T_1}u_i^\pone,0\right\}$, $l_2'=\max\left\{l_2-u_0^\ptwo -\sum_{i\in T_2}u_i^\ptwo ,0\right\}$, respectively.
\end{enumerate}

\begin{lemma}\label{smallcoro}
In the reduced problem $\mathcal{P}'$, every facility location $i\in F_1'$ covers at most $\frac{\epsilon}{\gamma}\left(|C_1\setminus C_1'|-u_0^\pone\right)$ clients with radius $3R^\star$, and every facility location $i\in F_2'$ covers at most $\frac{\epsilon}{\gamma}\left(|C_2\setminus C_2'|-u_0^\ptwo \right)$ clients with radius $3R^\star$.
\end{lemma}
\begin{proof}
First, if $i\in F_1'$ is co-located with some location in $T_1\cup h(T_2)$, the number of clients it covers in $C_1'$ is obviously zero, since they are all removed from $C_1'$. Otherwise, according to our construction of $F_1'$, for any $i\in F_1'$ that are not co-located with $T_1\cup h(T_2)$, it must cover at most $\min_{i\in T_1}\{u_i^\pone\}$ remaining clients in $C_1'$ with radius $3R^\star$, otherwise it would have been removed during the construction of $\calP'$. Of course, this number is also smaller than the average of all $u_i^\pone$'s, which is exactly
\[\frac{\epsilon}{\gamma}\sum_{i\in T_1}u_i^\pone=\frac{\epsilon}{\gamma}\left(|C_1\setminus C_1'|-u_0^\pone\right),\]
where we recall $C_1\setminus C_1'$ contains all clients covered by $T_1\cup h(T_2)$ and $|T_1|=\gamma/\epsilon$. The case with $F_2'$ is identical.
\end{proof}

\begin{lemma}\label{lp-guess-top}
There exists a guess $(T_1,T_2,g,h)$ such that the original instance reveals a solution $(A_1\subseteq F_1',A_2\subseteq F_2')$ with objective at most $R^\star$, satisfying $T_1\cup h(T_2)\subseteq A_1,\,T_2\cup g(T_1)\subseteq A_2$ in terms of multi-sets, facility at $i\in T_1$ moves to $g(i)\in g(T_1)$ and facility at $i'\in T_2$ comes from $h(i')\in h(T_2)$. Moreover, there are at most $O\left(\left(|F_1|\cdot|F_2|\right)^{2\gamma/\epsilon}\right)$ such different guesses.
\end{lemma} 
\begin{proof} The number of possible guesses is easy to see from the definition. Consider the optimal solution $(U_1\subseteq F_1,U_2\subseteq F_2)$ to the \emph{original problem}. Choose $\gamma/\epsilon$ facilities $U_1'\subseteq U_1$, such that the total number of clients they cover in radius $3R^\star$ is maximized. Define $U_2'\subseteq U_2$ similarly. Assume that we have made the correct guesses $T_1=U_1',\,T_2=U_2'$ (see that their cardinalities are the same), and our guesses $g$ and $h$ correctly depict their movement in the optimal solution.

Now we only need to prove that $U_1\subseteq F_1'$ and $U_2\subseteq F_2'$. Since $U_1'=T_1$, we can sort $U_1'$ in the same way as $T_1$, and it is easy to see that any $i'\in U_1\setminus U_1'$ can only cover at most $\min_{i\in T_1}\{u_i^\pone\}$ new clients in $C_1\setminus\bigcup_{i\in T_1}B_1(i,3R^\star)$, otherwise we may always replace the last one in $U_1'$ with $i'$ and cover more clients, a contradiction. 
Since $C_1'$ is a subset of $C_1\setminus\bigcup_{i\in T_1}B_1(i,3R^\star)$, each $i'\in U_1\setminus U_1'$ is \emph{not} removed from $F_1$ when $F_1'$ is created, thus one has $U_1\setminus U_1'\subseteq F_1'$ and $U_1\subseteq F_1'$ as a result. The same argument also shows $U_2\setminus U_2'\subseteq F_2'$ and $U_2\subseteq F_2'$. We recover the lemma by setting $A_1=U_1$ and $A_2=U_2$.
\end{proof}

\begin{remark}
We remark that, while we allow $g(T_1)$ and $h(T_2)$ to induce multi-sets, we only consider the cases where $T_1\subseteq F_1$ and $T_2\subseteq F_2$ are subsets. We briefly explain the reason here. From the proof of the above Lemma \ref{lp-guess-top}, one can easily see that if $U_1$ in the optimum (as a multi-set itself) has less than $\gamma/\epsilon$ distinct members of $F_1$, we can simply guess these members, delete every other facility location that is not co-located with them, and still have an optimal solution with objective $R^\star$ to the remaining instance. This process takes at most $|F_1|+|F_1|^2+\cdots+|F_1|^{\gamma/\epsilon}=O(|F_1|^{1+\gamma/\epsilon})$ guesses in total, and the remaining instance is much easier to solve, since the guessed distinct facilities in $U_1$ should already cover at least $l_1$ clients with radius $R^\star$. Therefore, we omit these scenarios here for simplicity.
\end{remark}

From now on, further assume that we have made the correct guess $(T_1,T_2,g,h)$ as shown in Lemma \ref{lp-guess-top} and reached the reduced instance $\calP'$. We define the following natural LP relaxation. By adding a superscript to every variable to indicate the time step, denote $x_{ij}^\pt\in[0,1]$ the partial assignment of client $j$ to facility $i$ and $y_i^\pt\geq0$ the extent of opening facility location $i$ at time step $t$. Moreover, denote $z_{ii'}$ the extent of movement from facility $i$ to facility $i'$, between neighboring time steps $t=1$ and $t=2$.
\begin{alignat*}{3}
\sum_{j\in C_t'}\sum_{i\in F_t',d(i,j)\leq R^\star}x_{ij}^\pt&\geq l_t'&\quad&\forall t=1,2\tag{$\mathrm{LP}(\calP')$}\label{lp-center-reduced}\\
	\sum_{i\in F_t',d(i,j)\leq R^\star}x_{ij}^\pt&\leq 1&&\forall j\in C_t',\,t=1,2\\
	\sum_{i\in F_t'}y_i^\pt&=k&&\forall t=1,2\\
	0\leq x_{ij}^\pt&\leq  y_i^\pt&&\forall i\in F_t',j\in C_t',\,t=1,2\\
	x_{ij}^\pt&=0&&\forall i\in F_t',j\in C_t',d(i,j)>R^\star,\,t=1,2\\
	\sum_{i\in F_1',d(i,i')\leq B}z_{ii'}&=y_{i'}^\ptwo &&\forall i'\in F_2'\\
	\sum_{i'\in F_2',d(i,i')\leq B}z_{ii'}&=y_{i}^\pone&&\forall i\in F_1'\\
	z_{ii'}&\geq 0&&\forall i\in F_1',i'\in F_2'\\
	z_{ii'}&=0&&\forall i\in F_1',i'\in F_2',d(i,i')>B\\
	z_{ig(i)}&\geq1&&\forall i\in T_1\\
	z_{h(i')i'}&\geq1.&&\forall i'\in T_2
\end{alignat*}

\begin{lemma}
\ref{lp-center-reduced} is feasible.
\end{lemma}

\begin{proof}
	Consider the optimal solution $U_1\subseteq F_1,U_2\subseteq F_2$, and define the variables ${x}_{ij}^{\star(t)},{y}_i^{\star(t)}$ and ${z}_{ii'}^{\star}$ which are restricted to the reduced instance $\calP'$ accordingly. Using Lemma \ref{lp-guess-top}, if $F_1',F_2'$ is computed according to $T_1$ and $T_2$, we have $U_1\subseteq F_1',U_2\subseteq F_2'$, and it is easy to check that all but the first constraint of \ref{lp-center-reduced} are satisfied by $(x^\star,y^\star,z^\star)$.
	
	Now consider the first constraint. In the optimal solution, $U_1$ covers at least $l_1$ clients in $C_1$ with radius $R^\star$ and $T_1\subseteq U_1,h(T_2)\subseteq U_1$. The facilities in $T_1\cup h(T_2)$ cover exactly $u_0^\pone+\sum_{i\in T_1}u_i^\pone$ clients with radius $3R^\star>R^\star$, which are all included in $C_1\setminus C_1'$. Evidently, for those in $C_1'$ that are \emph{not} covered by the previous larger balls, the remaining facilities in $U_1\setminus(T_1\cup h(T_2))$ have to cover at least $\max\{l_1-u_0^\pone-\sum_{i\in T_1}u_i^\pone,0\}=l_1'$ of them with radius $R^\star$. The same holds for $C_2'$ and $l_2'$, thus the first constraint is also satisfied by $(x^\star,y^\star,z^\star)$, and \ref{lp-center-reduced} is feasible.
\end{proof}

\subsubsection{Matching-based LP rounding}
\label{section:dksout:matching}

Given a fractional solution $(x,y,z)$ to \ref{lp-center-reduced}, define $s_j^\pt=\sum_{i\in F_t',d(i,j)\leq R^\star}x_{ij}^\pt$ the extent of connection of client $j\in C_t'$. We conduct the standard filtering algorithm (see, e.g.,~\cite{harris2019lottery,chakrabarty2019generalized}) to filter the clients, defined as in Algorithm \ref{algo:greedy}. According to the definition of the algorithm, for any $j\in C_t'$, there exists some $j'\in C_t''$ such that $d(j,j')\leq 2R^\star$. Therefore, if we can cover all the clients in $C_t''$ with a radius of at most $\kappa R^\star$, using triangle inequality, every client in $C_t'$ can be covered within radius $(\kappa+2) R^\star$. 

\begin{algorithm}[hbt!]
\caption{GREEDYFILTER}\label{algo:greedy}
\DontPrintSemicolon
\SetKwInOut{Input}{Input}
\SetKwInOut{Output}{Output}

\Input{$\calP',R^\star,(x,y,z)$}
\Output{two subsets of filtered clients for $t=1,2$, with each client having a certain profit value}
\For{$t=1,2$}{
$C_t''\leftarrow\emptyset$\;
\For{unmarked $j\in C_t'$ in non-increasing order of $s_j^\pt$}{
$C_t''\leftarrow C_t''\cup\{j\}$\;
set each \emph{unmarked} $j'\in C_t'$ s.t. $d(j,j')\leq 2R^\star$ as marked\;
let $c_j^\pt$ be the number of clients marked in this iteration
\tcp*{the ``profit'' of $j\in C_t''$}
}
$c^\pt\leftarrow\left(c_j^\pt:j\in C_t''\right)$\;}
\Return $(C_t'',c^\pt),\,t=1,2$\;
\end{algorithm}

The following two lemmas illustrate the relative sparsity of $C_1'$ around (almost) each client $j\in C_1''$ and provide an inequality which will be useful for achieving our final approximate solution. The lemmas for $t=2$ are the same, hence omitted here.

\begin{lemma}\label{lemma:sparse}
For any $j\in C_1''$ with $s_j^\pone>0$, $c_j^\pone\leq\frac{\epsilon}{\gamma}\left(|C_1\setminus C_1'|-u_0^\pone\right)$.
\end{lemma}
\begin{proof}
Evidently, $c_j^\pone$ is at most the number of clients in $C_1'$ that are $\leq 2R^\star$ away from $j$. Since $s_j^\pone>0$, there exists $i'\in F_1'$ with $d(i',j)\leq R^\star$, and these $c_j^\pone$ clients are at most $3R^\star$ away from $i'$ using triangle inequality. One also has that $i'$ is not co-located with $T_1\cup h(T_2)$, since otherwise $j$ would have been removed from $C_1$ during pre-processing, which is a contradiction. Using Lemma \ref{smallcoro}, $c_j^\pone$ is at most the number of clients in $C_1'$ covered by $i'$ with radius $3R^\star$, hence at most
$\left(|C_1\setminus C_1'|-u_0^\pone\right)\epsilon/\gamma$.
\end{proof}

We note that if $s_j^\pone=0$, the lemma above does not hold. Nevertheless, our algorithm handles such clients by simply ignoring them in the rounding process.

\begin{lemma}\label{lemma:outlierguarantee}
	$\sum_{j\in C_1''}c_j^\pone s_j^\pone\geq l_1'.$
\end{lemma}
\begin{proof}
In Algorithm \ref{algo:greedy}, the subsets are processed in non-increasing order of $s_j^\pone$, and $c_j^\pone$ denotes the number of clients that are marked in each step, hence one has
\[\sum_{j\in C_1''}c_j^\pone s_j^\pone\geq\sum_{j\in C_1'}s_j^\pone\geq l_1',\]
where we recall that $s_j^\pone=\sum_{i\in F_1',d(i,j)\leq R^\star}x_{ij}^\pone$, and the last inequality is due to the feasibility of $x_{ij}^\pone$ for \ref{lp-center-reduced}.
\end{proof}

We first use the following Algorithm \ref{algo:split} to modify the LP solution. The main goal of this algorithm is to duplicate, and in some cases, merge facilities, such that the modified solution $(\hat y, \hat z)$ is easier to round on the new instance.

\begin{algorithm}[hbt!]
\caption{SPLIT\&MERGE}\label{algo:split}
\DontPrintSemicolon
\SetKwInOut{Input}{Input}
\SetKwInOut{Output}{Output}

\Input{$\calP',R^\star,(x,y,z),C_1'',C_2'',c^\pone,c^\ptwo ,T_1,T_2,g,h$}
\Output{a modified instance of 2-\rdks and a feasible fractional solution}

$F_1''\leftarrow F_1',F_2''\leftarrow F_2',\hat y\leftarrow y,\hat z\leftarrow z$\;
\For{each distinct pair $(i,i')\in\{(i,g(i)):i\in T_1\}\cup\{(h(i'),i'):i'\in T_2\}$}{
create co-located copies $i_1,i_1'$ of $i,i'$, respectively, $F_1''\leftarrow F_1''\cup\{i_1\},F_2''\leftarrow F_2''\cup\{i_1'\}$\;
$\hat y_{i_1}^\pone\leftarrow1,\hat y_{i_1'}^\ptwo \leftarrow1,\hat z_{i_1i_1'}\leftarrow1$, $\hat y_i^\pone\leftarrow \hat y_i^\pone-1,\hat y_{i'}^\ptwo \leftarrow \hat y_{i'}^\ptwo -1,\hat z_{ii'}\leftarrow \hat z_{ii'}-1$\;}
\For{$t=1,2$ and $j\in C_t''$}{
$E_j^\pt\leftarrow\{i\in F_t'':d(i,j)\leq R^\star\}$\;
\While{there exists $i\in E_j^\pt$ s.t. $\hat y_i^\pt>x_{ij}^\pt$}{
split $i$ into co-located copies $i_1,i_2$, $\hat y_{i_1}^\pt\leftarrow x_{ij}^\pt$, $\hat y_{i_2}^\pt\leftarrow\hat y_i^\pt-\hat y_{i_1}^\pt$, split corresponding movement variables $\hat z$ accordingly\;
$F_t''\leftarrow \left(F_t''\setminus\{i\}\right)\cup\{i_1,i_2\}$, $E_j^\pt\leftarrow \left(E_j^\pt\setminus\{i\}\right)\cup\{i_1\}$\;}
merge all locations in $E_j^\pt$ into a single one $f_j^\pt$, $\hat y_{f_j^\pt}^\pt\leftarrow s_j^\pt$, merge corresponding $\hat z$ values accordingly\;\label{mergecluster}
$F_t''\leftarrow \left(F_t''\setminus E_j^\pt\right)\cup\{f_j^\pt\}$\;}
\While{there exists $i\in F_t''$ s.t. $\hat y_{i}^\pt>1$}{
split $i$ into $\lceil \hat y_i^\pt\rceil$ co-located copies, such that the first $\lfloor \hat y_i^\pt\rfloor$ of them have $\hat y^\pt$ values 1, and the last one (if any) has $\hat y^\pt$ value $\hat y_i^\pt-\lfloor \hat y_i^\pt\rfloor$, split corresponding movement variables $\hat z$ accordingly\;}
\Return $(\hat y,\hat z),F_1'',F_2''$\;
\end{algorithm}

In Algorithm \ref{algo:split}, the first loop is used to reserve the variables for our pre-selected facilities $T_1,T_2$ and $g(T_1),h(T_2)$. We note that this is always possible, since in \ref{lp-center-reduced}, we explicitly set $z_{ig(i)}\geq1$ and $z_{h(i')i'}\geq1$ for $i\in T_1$ and $i'\in T_2$, which also implicitly sets up lower bounds for relevant $y$ variables. Secondly, we define a subset $E_j^\pt$ of facilities that are close to $j\in C_t'',t=1,2$. Using the standard facility duplication technique, we make sure that $\forall i\in E_j^\pt$, one has $\hat y_i^\pt=x_{ij}^\pt$, which in turn implies that $\hat y^\pt(E_j^\pt)=s_j^\pt$. This even holds for $j\in C_t''$ with $s_j^\pt=0$. Using this property, we simply merge the facilities in $E_j^\pt$ into a single one, called $f_j^\pt$, and assign it $\hat y_{f_j^\pt}^\pt\leftarrow s_j^\pt$ (recall that $F_1''$ and $F_2''$ live in \emph{different} metric spaces now, so this merging process has no effect on the other time step). We also clarify that, whenever a facility location $i$ is split into multiple copies, we always need to respect the fractional movement variables $\hat z$ and make sure that they are also properly split in order to satisfy the ``conservation'' of fractional facilities. Here, we do not need any particular splitting criteria, and simply use an arbitrary one.

From now on, we view the facilities in $F_1''$ and $F_2''$ as vertices in a bipartite graph, and the variables in $\hat z$ as a fractional matching on it. Roughly speaking, if some $\hat z_{ii'}$ is rounded to one, it means that a facility is moved from $i\in F_1''$ to $i'\in F_2''$. Further, if $i$ is equal to some merged facility $f_j^\pone$ with $j\in C_1''$, it means that some facility in $E_j^\pone$ (before merging) can be opened, and we can use this facility to cover at least $c_j^\pone$ distinct clients in $C_1'$ with radius $3R^\star$, using triangle inequality and the fact that any client $j'$ marked by $j$ satisfies $d(j,j')\leq 2R^\star$. This observation motivates us to assign a ``profit'' value to each merged facility, and define a multi-objective optimization problem in the following lemma.

\begin{lemma}\label{lemma:matching}
For the modified LP solution $(\hat y,\hat z)$ by Algorithm \ref{algo:split}, create a bipartite graph $G=\left(V_1\cup V_2,E\right)$, where $V_1,V_2$ represent the entries in $F_1'',F_2''$, respectively, and edge $(i,i')\in E$ is defined for every non-zero $\hat z_{ii'}>0$. $\hat z$ is a fractional $k$-cardinality bipartite matching over $G$.
	
Moreover, assign profit $p_1(e)=c_j^\pone$ to every edge $e$ incident on merged facility $f_j^\pone,j\in C_1''$, profit $p_2(e')=c_{j'}^\ptwo $ to every edge $e'$ on merged facility $f_{j'}^\ptwo ,j'\in C_2''$, and set all other profit values to zero. $\hat z$ satisfies the following ``profitability constraints'', \[P_1(\hat z)=\sum_{e\in E}\hat z_ep_1(e)\geq l_1',\,P_2(\hat z)=\sum_{e\in E}\hat z_ep_2(e)\geq l_2'.\]
\end{lemma}
\begin{proof}
	In Algorithm \ref{algo:split}, for every merged facility $f_j^\pt$, it satisfies $\hat y_{f_j^\pt}^\pt=s_j^\pt\leq 1$. For other facilities, we split those with $\hat y_i>1$ into co-located copies, and every copy is matched up to an extent of 1. The total extent of matched edges is directly from the constraint $\sum_{i\in F_1'}y_i^\pone=k$ in \ref{lp-center-reduced}. Hence $\hat z$ is a fractional $k$-cardinality matching on $G$.
	
	To see that the profitability constraints are satisfied, we focus on $t=1$. Since each edge that is not incident on any $f_j^\pone,\,j\in C_1''$ has $p_1(e)=0$, we may rewrite the total profit as,
	\[P_1(\hat z)=\sum_{j\in C_1''}\sum_{e=\left(f_j^\pone,i''\right)}\hat z_ep_1(e)=\sum_{j\in C_1''}c_j^\pone\left(\sum_{e=\left(f_j^\pone,i''\right)}\hat z_e\right)=\sum_{j\in C_1''}c_j^\pone \hat y_{f_j^\pone}^\pone=\sum_{j\in C_1''}c_j^\pone s_j^\pone\geq l_1',\]
	where the last inequality is due to Lemma \ref{lemma:outlierguarantee}.
\end{proof}

We also notice that if $s_j^\pone=0,j\in C_1''$, the merged facility satisfies $\hat y_{f_j^\pone}^\pone=0$ and every associated movement variable also satisfies $\hat z_{f_j^\pone i'}=0$. This means that no edges are incident on $f_j^\pone$, and we can simply remove such vertices from the graph. With the feasibility lemma above, we present the following main theorem, in which we round the aforementioned fractional bipartite matching on $G$ to an integral matching, and directly obtain a solution to the original 2-\rdks instance.
\begin{theorem}\label{thm:dksout}
	For any constant $\epsilon>0$, there exists a multi-criteria $(3,1-\epsilon,1-\epsilon)$-approximation for \emph{2-\rdks}.
\end{theorem}

\begin{proof} We consider the fractional matching induced by $\hat z$ in Lemma \ref{lemma:matching} and adapt the rounding procedures in Section~4 of~\cite{grandoni2014new}, with some necessary modifications. First, since Algorithm \ref{algo:split} explicitly sets $\hat z_{ig(i)}=1$ and $\hat z_{h(i')i'}=1$ for each $i\in T_1,\,i'\in T_2$, these edges are always matched, and we can remove them in advance. Suppose there are $\kappa\geq0$ such edges, and we remove these edges and their endpoints. 

Notice that $p_1(e)=p_2(e)=0$ if $e$ is removed in this process. To see this, both endpoints of $e$ are in $T_1\cup h(T_2)$ and $T_2\cup g(T_1)$, and every client within $3R^\star$ from them is already removed from $C_1$ and $C_2$. W.l.o.g., assume $p_1(e)>0$ for some $e$ removed, which means that there exists $j\in C_1''$ such that $E_j^\pone$ contained some $i\in T_1\cup h(T_2)$ before being merged. This puts $d(i,j)\leq R^\star$, and we would have removed $j$ from $C_1$ during pre-processing, a contradiction.

Let $P_\calM$ be the matching polytope of the remaining graph, and consider the following LP,
\begin{equation}
    \left\{\min \mathbbm{1}^\top z:z\in P_\calM,\,P_1(z)=\sum_{e\in E}z_ep_1(e)\geq l_1',\,P_2(z)=\sum_{e\in E}z_ep_2(e)\geq l_2'\right\}.
\end{equation}

According to Lemma \ref{lemma:matching}, the solution $\hat z$ given by Algorithm \ref{algo:split} is feasible to the above LP with objective $k-\kappa$ (after the removal of those must-have edges, with each of them having zero profits), hence the optimum is at most $k-\kappa$. Let $z_0$ be such an optimal basic solution with objective $\mathbbm{1}^\top z_0\leq k-\kappa$.

Now that $z_0$ is a basic solution, it lies on a face of $P_\mathcal{M}$ of dimension at most $2$, thus using Carath\'eodory theorem, it can be written as the convex combination of 3 (integral) basic solutions of $P_\calM$, say $z_0=\alpha_1z_1+\alpha_2z_2+\alpha_3z_3$, where $\alpha_i\in[0,1],\alpha_1+\alpha_2+\alpha_3=1$, and $z_1,z_2,z_3$ are three basic solutions to $P_\calM$. W.l.o.g., we assume that $\alpha_1,\alpha_2,\alpha_3$ are all positive numbers.

\begin{description}
\item[Construction of an intermediate matching $\hat z_2$.]
We create an almost-matching that fractionally combines $z_1,z_2$. To be more precise, let $z_{1,2}=\frac{\alpha_1}{\alpha_1+\alpha_2}z_1+\frac{\alpha_2}{\alpha_1+\alpha_2}z_2$ be the convex combination of $z_1$ and $z_2$ thus a fractional matching, we want to find another almost-integral matching $z_2'\in[0,1]^E$ s.t.
\begin{equation}
P_1(z_2')=P_1(z_{1,2}),P_2(z_2')=P_2(z_{1,2}),\mathbbm{1}^\top z_2'=\mathbbm{1}^\top z_{1,2},\label{eq:1st:matching}
\end{equation}
and it is possible to set at most 4 variables in $z_2'$ to 0 and obtain an integral matching. Using Corollary~4.10 in~\cite{grandoni2014new}, such an almost-matching $z_2'$ exists and can be efficiently computed. We then set at most 4 variables in $z_2'$ to 0 and obtain a matching $\hat z_2$. It is obvious that
\begin{equation}P_1(\hat z_2)\geq P_1(z_2')-4\max_{e\in E}p_1(e),\,P_2(\hat z_2)\geq P_2(z_2')-4\max_{e\in E}p_2(e),\,\mathbbm{1}^\top \hat z_2\leq\mathbbm{1}^\top z_2'.\label{eq:2nd:matching}
\end{equation}
\item[Construction of the final matching $\hat z_3$.]
Using Corollary~4.10 in~\cite{grandoni2014new} again, let $z_{2,3}=(\alpha_1+\alpha_2)\hat z_2+\alpha_3z_3$, we can efficiently find $z_3'\in[0,1]^E$ s.t.
\begin{equation}
P_1(z_3')=P_1(z_{2,3}),P_2(z_3')=P_2(z_{2,3}),\mathbbm{1}^\top z_3'=\mathbbm{1}^\top z_{2,3},\label{eq:3rd:matching}
\end{equation}
and it is possible to set at most 4 variables in $z_3'$ to 0 and obtain a matching $\hat z_3$. It is obvious that
\begin{equation}P_1(\hat z_3)\geq P_1(z_3')-4\max_{e\in E}p_1(e),\,P_2(\hat z_3)\geq P_2(z_3')-4\max_{e\in E}p_2(e),\,\mathbbm{1}^\top \hat z_3\leq\mathbbm{1}^\top z_3'.\label{eq:4th:matching}
\end{equation}

Combining~\eqref{eq:1st:matching}\eqref{eq:2nd:matching}\eqref{eq:3rd:matching}\eqref{eq:4th:matching}, one has
\begin{align*}
	P_1(\hat z_3)&\geq P_1(z_3')-4\max_{e\in E}p_1(e)= P_1(z_{2,3})-4\max_{e\in E}p_1(e)\\
	&= (\alpha_{1}+\alpha_2)P_1(\hat z_2)+\alpha_3P_1(z_3)-4\max_{e\in E}p_1(e)\\
	&\geq  (\alpha_{1}+\alpha_2)\left(P_1(z_2')-4\max_{e\in E}p_1(e)\right)+\alpha_3P_1(z_3)-4\max_{e\in E}p_1(e)\\
	&\geq  (\alpha_{1}+\alpha_2)P_1(z_{1,2})+\alpha_3P_1(z_3)-8\max_{e\in E}p_1(e)\\
	&=P_1(\alpha_1z_1+\alpha_2z_2+\alpha_3z_3)-8\max_{e\in E}p_1(e)\geq l_1'-\frac{8\epsilon}{\gamma}\left(|C_1\setminus C_1'|-u_0^\pone\right),
\end{align*}
where the last inequality is due to Lemma \ref{lemma:sparse} and the fact that any non-zero profit must be defined for an edge incident on $f_j^\pone,\,j\in C_1''$ with $0<\hat y_{f_j^\pone}^\pone=s_j^\pone$. Similarly, one has $P_2(\hat z_3)\geq l_2'-\frac{8\epsilon}{\gamma}\left(|C_2\setminus C_2'|-u_0^\ptwo \right)$ and easily sees $\mathbbm{1}^\top \hat z_3\leq \mathbbm{1}^\top z_3'=\mathbbm{1}^\top z_{2,3}=(\alpha_1+\alpha_2)\mathbbm{1}^\top \hat z_2+\alpha_3\mathbbm{1}^\top z_3\leq(\alpha_1+\alpha_2)\mathbbm{1}^\top z_{1,2}+\alpha_3\mathbbm{1}^\top z_3=\mathbbm{1}^\top z_0\leq k-\kappa$. Hence the cardinality of $\hat z_3$ is at most $k-\kappa$.

\item[Construction of the output solution.]
Let $\gamma=8$, $M$ be the set of edges matched in $\hat z_3$ plus the $\kappa$ edges removed in the beginning, and $A_1,A_2$ be two multi-sets that are initially empty. We have $|M|\leq k$. For each $e\in M$, consider the following cases.
\begin{description}
    \item[$e$ is among the $\kappa$ edges removed in the beginning.] Add the two endpoints to $A_1$ and $A_2$, respectively.
    \item[$e$ is not incident on any merged facility.] Add the two endpoints to $A_1$ and $A_2$, respectively.
    \item[$e$ is incident on exactly one merged facility.] W.l.o.g., let this endpoint be $f_j^\pone,j\in C_1''$ and the other be $i'\in F_2''$. Notice that according to Lemma \ref{lemma:matching} and the merging process in Algorithm \ref{algo:split}, $e\in E$ indicates that $\hat z_{f_j^\pone i'}>0$ when we construct $G$ in the first place, thus there must exist $i\in E_j^\pone$ s.t. $\hat z_{ii'}>0$ before merging and hence $d(i,i')\leq B$. Add $i$ to $A_1$ and $i'$ to $A_2$.
    \item[$e$ is incident on two merged facilities.] Let them be $f_j^\pone,j\in C_1''$ and $f_{j'}^\ptwo,j'\in C_2''$. Again, $e\in E$ implies that $\hat z_{f_j^\pone f_{j'}^\ptwo}>0$, thus there must exist $i\in E_j^\pone$ and $i'\in E_{j'}^\ptwo$ s.t. $\hat z_{ii'}>0$ and hence $d(i,i')\leq B$. Add $i$ to $A_1$ and $i'$ to $A_2$.
\end{description}
\end{description}

The construction of $(A_1,A_2)$ naturally induces a feasible perfect matching between them, such that any matched pair is no more than $B$ away from each other. We already know $(T_1\cup h(T_2))\subseteq A_1$, and it covers all clients in $C_1\setminus C_1'$ with radius $3R^\star$. From the value of $P_1(\hat z_3)$, at least $l_1'-\epsilon\left(|C_1\setminus C_1'|-u_0^\pone\right)$ distinct clients in $C_1'$ are covered additionally with radius $3R^\star$, so the total number of clients covered by $A_1$ (using radius $3R^\star$) is at least
\begin{align*}
P_1(\hat z_3)+|C_1\setminus C_1'|&\geq l_1'-\epsilon\left(|C_1\setminus C_1'|-u_0^\pone\right)+u_0^\pone+\sum_{i\in T_1}u_i^\pone\\
&\geq \max\left\{l_1-u_0^\pone-\sum_{i\in T_1}u_i^\pone,0\right\}+u_0^\pone+\sum_{i\in T_1}u_i^\pone-\epsilon\left(|C_1\setminus C_1'|-u_0^\pone\right)\\
&\geq l_1-\epsilon l_1=(1-\epsilon)l_1,
\end{align*}
where we assume that $|C_1\setminus C_1'|-u_0^\pone<l_1$, otherwise we would have already covered $\geq l_1$ clients using only $T_1$. The proof is the same for $A_2$ and the second time step. Arbitrarily add more facilities to $A_1$ and $A_2$ until $|A_1|=|A_2|=k$, and this gives our final multi-criteria $(3,1-\epsilon,1-\epsilon)$-approximate solution.
\end{proof}

\section{Future work}\label{section:future}

It would be very interesting to remove 
the dependency of $\gamma$ (the coefficient of movement cost) and $\epsilon$ (the lower bound of the weight) from the approximation factor for \dom in Theorem~\ref{thm:dom}, or to show that such dependency is inevitable.
We leave it as an important open problem.
We note that a constant approximation factor for \dom without depending on $\gamma$ would imply a constant approximation for stochastic $k$-server, for which only a logarithmic-factor approximation algorithm is previously known \cite{dehghani2017stochastic}.

Our approximation algorithm for \dom is based on the technique developed by Aouad and Segev \cite{aouad2019ordered} and Byrka~\etal\cite{byrka2018constant}. 
The original ordered $k$-median problem has subsequently seen improved approximation results in \cite{chakrabarty2018interpolating,chakrabarty2019approximation}. 
We did not try hard to optimize the constant factors.
Nevertheless, it is an interesting future direction to improve the constant factors by leveraging new techniques and ideas.
On the other hand, it would also be very interesting to obtain better lower bounds than the trivial lower bounds of $k$-supplier \cite{hochbaum1986unified} and $k$-median \cite{jain2002greedy}.

It is possible to formulate other problems that naturally fit into the dynamic clustering theme and are well motivated by realistic applications.
Besides the obvious variants of \dom and \dks via changing the clustering and movement objectives, we list some examples that are of particular interest.
\begin{itemize}
	\item In each time step, the constraint of opening at most $k$ facilities can be replaced by a matroid or knapsack constraint; that is, we require the open facilities to form an independent set of a given matroid, or to have a total weight no more than a given threshold.
		These formulations generalize the $k$-clustering setting we adopt in this paper.
		Many such clustering problems are studied in the literature, e.g., matroid center \cite{chen2016matroid}, matroid median \cite{swamy2016improved,krishnaswamy2018constant,krishnaswamy2011matroid}, knapsack center \cite{hochbaum1986unified} and knapsack median \cite{swamy2016improved,krishnaswamy2018constant,kumar2012constant,gupta2021structural}, and we can study the dynamic versions of these problems.
	\item We can consider fair dynamic clustering problems. 
		For example, at each time step $t$, the clients have different colors representing the demographics they belong to, and the goal is to cluster the clients such that for each $t$, the proportions of different colors in each cluster are roughly the same as their global proportions at time $t$ (see, e.g., \cite{chierichetti2017fair,bera2019fair,huang2019coresets}).
	\item One may also study other clustering criteria with additional constraints under the dynamic setting, e.g., the fault-tolerant versions \cite{hajiaghayi2016constant,khuller2000fault,swamy2008fault} and the capacitated versions \cite{khuller2000capacitated,chuzhoy2005approximating,li2016approximating}.
\end{itemize}

\section*{Acknowledgements}

Shichuan Deng and Jian Li were supported by the National Natural Science Foundation of China Grant 61822203, 61772297, 61632016, 61761146003, the Zhongguancun Haihua Institute for Frontier Information Technology, Turing AI Institute of Nanjing, and Xi'an Institute for Interdisciplinary Information Core Technology. 
Yuval Rabani was supported by ISF grant number 2553-17.

We thank Chaitanya Swamy for kindly pointing out studies relevant to our results. 
We also thank the anonymous referees for their insightful and constructive comments.


\bibliographystyle{plainurl}
\bibliography{references.bib}

\end{document}